\documentclass[11pt]{article}

\usepackage{a4wide}

\usepackage{amsmath,amssymb,amsthm,amscd}

\usepackage[all]{xy}

\newcommand{\bb}[1]{\begin{equation}\label{#1}}
\newcommand{\ee}{\end{equation}}
\newcommand{\vc}[1]{{\bf #1 }}

\newcommand{\Sum}[1]{\sum\limits_{#1}}
\newcommand{\Int}[1]{\int\limits_{#1}}

\newcommand{\ol}{\overline}
\newcommand{\tr}{\operatorname{tr}}
 
\newcommand{\bbR}{\mathbb R} 

\newcommand{\bbZ}{\mathbb Z}

\newcommand{\g}{{\bf g}}

\newcommand{\lpartial}{\overset{\rightarrow}{\partial}}

\newcommand{\lder}[1]{\overset{\rightarrow}{\vc #1}}

\newcommand{\supp}{\operatorname{supp}}

\newcommand{\E}{\mathcal E} 
\newcommand{\M}{\mathcal M}

\def\bothidenty{\rlap{\hbox to.97\wd0{\hss\vrule height.06\ht0 width.82\wd0}}
 \copy0\rlap{\kern-.36\wd0\vrule height1.05\ht0 width.05\ht0}\kern.14\wd0}
\newtheorem{theorem}{Theorem}[section]

\newtheorem{definition}[theorem]{Definition}

\newtheorem{lemma}[theorem]{Lemma}

\newtheorem{proposition}[theorem]{Proposition}

\theoremstyle{definition}

\def\E{{\mathcal{E}}}

\def\M{{\mathcal{M}}}

\def\L{{\mathcal{L}}}

\def\S{{\mathcal{S}}}

\def\U{{\mathcal{U}}}

\def\Z{{\mathbb Z}}

\def\N{{\mathbb N}}

\newcommand{\bbM}{\mathbb M}

\def\id{{\mbox{id}}}

\begin{document}

%\title[Batalin--Vilkovisky Integrals in Finite Dimensions]
%{Batalin--Vilkovisky Integrals in Finite Dimensions}
\title{Batalin--Vilkovisky Integrals in Finite Dimensions}

\author{C. Albert\footnote{present address: Section de math\'ematiques, University of Geneva, CH--1211 Geneva 4},\,\, 
B. Bleile\footnote{present address: School of Science and Technology, University of New England, Armidale 2351, Australia }\,\, and J. Fr{\"o}hlich
\\
Theoretical Physics, ETH-H\"onggerberg, CH-8093 Z\"urich}

%\author{Carlo Albert\footnote{Section de math{e}matiques, University of Geneva, CH-1211 Gen{e}ve 4}, \foonote{Institut f{\"u}r Theoretische Physik, ETH Z{\"u}rich, CH--8093 Z{\"u}rich}}
 
%\address
%{Universit{'e} de Geneve \\
%ETH Z{\"u}rich\\
%CH--8093 Z{\"u}rich} 

%\address
%{Institut f{\"u}r Theoretische Physik \\
%ETH Z{\"u}rich\\
%CH--8093 Z{\"u}rich} 

%\address
%{Institut f{\"u}r Theoretische Physik \\
%ETH Z{\"u}rich\\
%CH--8093 Z{\"u}rich} 
\maketitle

\begin{abstract}
The Batalin-Vilkovisky method ($BV$) is the most powerful method to analyze functional integrals with (infinite-dimensional) gauge symmetries presently known. It has been invented to fix gauges associated with symmetries that do not close off-shell.
Homological Perturbation Theory is introduced and used to develop the integration theory behind $BV$ and to describe the $BV$ quantization of a Lagrangian system with symmetries.
Localization (illustrated in terms of Duistermaat-Heckman localization) as well as anomalous symmetries are discussed in the framework of $BV$.
\end{abstract}

%\tableofcontents

%\input{intro080401.tex}%%%%%%%%%%%%%%%%%%%%%%%%%%%%%%%%%%%%%%%%%%%%%%%%%%%

\section{Introduction}\label{intro}

In the Lagrangian approach to (quantum) physics, we are given a {\em space of fields}, $M$, and an {\em action}, $S_0$, thereon.
Typically, the space of fields is an infinite-dimensional space of sections of some field bundle over space-time.
For the following to hold not only at formal level, we assume $M$ to be a finite-dimensional smooth manifold without boundary.
Since, in this paper, we focus our attention on various {\em cohomology} theories connected with gauge-fixing, we should mention that aspects of {\em local} cohomology are absent if we assume $M$ to be finite-dimensional.
We assume the action $S_0$ to be smooth and complex-valued, with a non-negative imaginary part.

The equations $S_{0,i}=0$, where $S_{0,i}$ denotes the derivative of $S_0$ w.r.t. the $i$'th coordinate function on $M$, are called {\em equations of motion}, and their set of solutions, $\Sigma\subseteq M$, is called the {\em shell}\footnote{This name comes from field theory jargon, where one speaks of the "mass shell" and of "on-shell" conditions etc.}.
A major difficulty in the calculation of path-integrals in field theory originates in the presence of {\em local gauge symmetries}.
In general, infinitesimal local gauge symmetries do not form a Lie algebra; however, {\em on-shell}, they do. In our finite-dimensional setting, a symmetry is given by a linear subspace, $P\subseteq \Gamma(TM)$, of sections of the tangent bundle with the following properties:
\begin{enumerate}
\item For all $\vc X\in P$,
\bb{IntroP1}
\vc X(S_0)=0\,.
\ee
\item
There are tensors $T\in \Lambda^2P^*\otimes P$ and $E\in \Lambda^2P^*\otimes\Gamma(\Lambda^2TM)$, with
\bb{IntroP2}
[\vc X,\vc Y]=T(\vc X,\vc Y)+dS_0\neg E(\vc X,\vc Y)\,,\quad\text{for all}\quad \vc X\,,\,\vc Y\in P\,.
\ee
\end{enumerate}
Furthermore, we assume that $P$ is a finitely generated $C^\infty(M)$-module.
Since, on-shell, $P$ is closed w.r.t. the Lie bracket, it determines a foliation of the shell.
In infinite dimensions, the leaf space of the foliation, if a manifold, carries a natural symplectic structure and is the {\em classical phase space} of the theory (see \cite{Zuckerman}).

In field theories without anomalies, $M$ is equipped with a formal measure, $\Omega_0$, that respects the symmetry, i.e., $\operatorname{div}_{\Omega_0} \vc X=0$, for all $\vc X\in P$. {\em Quantizing} these theories means calculating {\em vacuum expectation values} of gauge-invariant functions, 
$$
f\in C^\infty(M)^P:=\{f\in C^\infty(M)\,|\,\vc X(f)=0\,,\vc X\in P\}\,,
$$ 
by means of the following formal path-integral:
\bb{IntroPathintegral}
<f>=\frac{\Int{M}fe^{iS_0/\hbar}\Omega_0}{\Int{M}e^{iS_0/\hbar}\Omega_0}\,.
\ee
Both numerator and denominator in (\ref{IntroPathintegral}) tend to diverge, due to the presence of the symmetry generated by $P$, and thus expression (\ref{IntroPathintegral}) needs to be replaced by a {\em gauge-fixed} version.
We must do this in such a way that the result is manifestly independent of the gauge chosen.
If the symmetry {\em closes} off-shell, i.e., if $E=0$ in (\ref{IntroP2}), this problem has been solved with the help of the so-called $BRST$-method (going back to Becchi, Rouet, Stora \cite{BRS} and Tyutin \cite{T}).
Here we provide a brief description of this method and illustrate it with an example at the end of this introduction.
The $BRST$-method extends $M$ to a graded manifold $\M$, by adding auxiliary even and odd fields {\em (Lagrange multipliers and ghosts)}, in such a way that all symmetry vector fields $\vc X\in P$ are encoded in {\em one odd vector field} $\mathcal X$ on $\M$, which, as an operator on functions on $\M$ ($BRST$-operator), squares to zero: $\mathcal X^2=0$. For this latter property to hold, ghosts must be introduced.
Then, in (\ref{IntroPathintegral}), $S_0$ is replaced by $S_0+\mathcal X(\Psi)$, for a suitably chosen {\em odd} function, $\Psi$, on $\M$, and integration is extended to $\M$. The so-called {\em gauge-fixing fermion} $\Psi$ is a function of fields, Lagrange multipliers and ghosts, and we call it "suitable" if integration w.r.t. the Lagrange multipliers and ghosts (the latter being a Berezin integral) yields a {\em finite} measure on the space of orbits of the symmetry $P$ acting on $M$.
Due to the fact that $S_0$ and $\Omega_0$ are invariant under the symmetry encoded in $\mathcal X$, (\ref{IntroPathintegral}) then turns into a functional on cohomology-classes of $\mathcal X$, which implies that a variation of the gauge-fixing fermion $\Psi$ does not alter the expectation value.

The $BRST$-method can also be formulated using {\em anti-fields}, which were originally introduced by Zinn-Justin \cite{ZJ}: Each field, ghost and Lagrange multiplier is paired with a field of opposite statistics (its anti-field). Mathematically, this means that $\M$ is extended to its odd cotangent bundle
\bb{IntroBVspace}
\E=\Pi T^*\M\,.
\ee
The action $S_0$ is now replaced by $S:=S_0+S_1$ with 
\bb{IntroS1}
S_1=\sum_i z^\dag_i\mathcal X(z^i)\,,
\ee
where the $z^i$ are coordinate-functions on $\M$ and the $z^\dag_i$ are corresponding anti-fields. (In infinite dimensions, the sum has to be replaced by an integral.)
The gauge-fixing is then encoded in the choice of a submanifold in $\E$ given by the equations
\bb{IntroL}
z^\dag_i=\frac{\partial\Psi(z)}{\partial z^i}\,,
\ee
where $\Psi=\Psi(z)$ is the gauge-fixing fermion introduced above.
As an odd cotangent bundle, $\E$ is equipped with the natural odd symplectic structure
$$
\omega=dz^i\wedge dz^\dag_i\,,
$$
and eqs (\ref{IntroL}) define a {\em Lagrangian submanifold}\footnote{A submanifold of half the dimension of $\E$ where $\omega$ vanishes.}, $\L$, in $\E$.
We obtain
$$
S_{\arrowvert\L}=S_0+\mathcal X(\Psi)
$$
and the vacuum expectation value (\ref{IntroPathintegral}) is replaced by
\bb{IntroPathintegralBV}
<f>=\frac{\Int{\L}fe^{iS/\hbar}\Omega_\L}{\Int{\L}e^{iS/\hbar}\Omega_\L}\,,
\ee
where $\L$ is suitably chosen; "suitable" again in the sense that it yields a finite measure on the space of orbits of the symmetry.
Note that the measure $\Omega_\L$ in (\ref{IntroPathintegralBV}) must be well-defined for an {\em arbitrary} choice of a Lagrangian submanifold $\L$.
Khudaverdian \cite{Khudaverdian} calls such measures {\em semidensities}, and Severa \cite{Severa} discovered their cohomological nature. We will give a definition of semidensities in Section 2.
For the time being, we just mention that, in order to define a semidensity, we need a measure, $\Omega$, on the even submanifold, or {\em body}, of $\E$, and $\Omega$ and hence the semidensity depend on the path-integral measure $\Omega_0$.
Note that the $BRST$-operator (or rather: an extension of the $BRST$-operator to functions on $\E$) is now given by $\{S_1,.\}$, where $\{.,.\}$ denotes the odd Poisson bracket corresponding to the odd symplectic structure on $\E$, and $\mathcal X^2=0$ is equivalent to $\{S_1,S_1\}=0$.
Since $S_0$ does not depend on any anti-fields, we also have that $\{S_0,S_0\}=0$. Furthermore, (\ref{IntroP1}) implies that $\{S_0,S_1\}=0$ and thus 
\bb{IntroCME}
\{S,S\}=0\,,
\ee
which is called {\em Classical Master Equation (CME)}.

This reformulation of the $BRST$-formalism allowed Batalin and Vilkovisky \cite{BV} to treat path-integrals with {\em open} symmetries (i.e., symmetries with a non vanishing $E$ in (\ref{IntroP2})), which is why the use of anti-fields in gauge-fixing path-integrals is nowadays called the {\em $BV$-method}.
In Section 3 we shall see, that the cohomology of the co-boundary operator $\{S_0,.\}$ is the restriction of functions (of fields and anti-fields) to the shell.
This is the reason why open symmetries can potentially be gauge-fixed with the $BV$-method. 
While it is not possible to encode an open symmetry in a co-boundary operator $\mathcal X$ on functions on $\M$, it is often possible to do so on $\E$.
That is, it is often possible to find an extension of the action $S_0$ by anti-field terms to an action $S=S_0+S_1$ that satisfies the $CME$ (\ref{IntroCME}) in such a way that $\delta=\{S,.\}$ encodes the symmetry, i.e.,
for all gauge-invariant functions $f\in C^\infty(M)^P$,
\bb{Introcond1}
\delta (\pi_0^*(f))=0\,,
\ee
where $\pi_0:\E\rightarrow M$ is a projection.
If the symmetry closes off-shell, the term $S_1$ is given by the $BRST$-operator according to (\ref{IntroS1}).
If it is open, $S_1$ contains terms of higher order in the anti-fields and $\{S_1,.\}$ does {\em not} square to zero.

In Section 3 we will see that, under a certain regularity condition on $S_0$, gauge-invariant functions whose (closed) support lies outside the shell are $\delta$-exact:
\bb{Introcond2}
\delta(\pi_0^*(f))=0\,,\quad\ol{\supp f}\cap \Sigma=\emptyset\Rightarrow \pi_0^*(f)\quad\text{is $\delta$-exact}\,.
\ee
Unless our system is {\em semi-classically exact}, (\ref{IntroPathintegralBV}) does {\em not} define a functional on $\delta$-cohomology classes.
Along with the measure $\Omega$ on the body of $\E$ comes a {\em third co-boundary operator}. In $BV$-language, the invariance of $\Omega$ under the symmetry reads
$$
\Delta_\Omega S_1=0\,,
$$
where $\Delta_\Omega$ is (up to signs) the second-order differential co-boundary operator $\partial^2/\partial z^i\partial z^\dag_i$ depending on $\Omega$.
In Section 3 we shall see that (\ref{IntroPathintegralBV}) defines a functional on cohomology classes of the operator
\bb{IntrodeltaBV}
\delta_{BV}:=\delta+i\hbar\Delta_\Omega\,,
\ee
which squares to zero if the so-called {\em Quantum Master Equation (QME)},
\bb{IntroQME}
\frac{1}{2}\{S,S\}-i\hbar\Delta_{\Omega}S=0\,,
\ee
is satisfied.
Furthermore, we shall see that this implies that the expectation value (\ref{IntroPathintegralBV}) is invariant under Hamiltonian variations of the Lagrangian submanifold $\L$ if $f$ is gauge-invariant.
Since Hamiltonian variations of $\L$ are associated with variations of the gauge-fixing classical symmetries survive quantization if the $QME$ has a solution.
If the measure $\Omega$ fails to be invariant under the symmetry-flow generated by $S$, that is, if $\Delta_\Omega S\neq 0$, we must add higher order $\hbar$-terms ({\em "counter terms"}) to $S$ to obtain a solution of the $QME$.
In Section 3, we will find cohomological obstructions to solving both the $CME$ and the $QME$.
Obstructions to solving the $QME$ are called {\em anomalies} since they prevent classical symmetries from being quantized.

From (\ref{Introcond2}) and (\ref{IntrodeltaBV}) we derive the statement
\bb{Introcond2q}
\delta(\pi_0^*(f))=0\,,\quad\ol{\supp f}\cap \Sigma=\emptyset\Rightarrow \pi_0^*(f)=\delta_{BV}\text{-exact}+\mathcal O(\hbar)\,,
\ee
which gives rise to the {\em perturbative expansion} (in powers of $\hbar$) of a path-integral around solutions of the equations of motion.
In general, such expansions do not converge. If they terminate at finite order, we say that the path-integral {\em localizes} on the shell. In Section 3, we will rederive Duistermaat-Heckman localization in the $BV$-formalism.

In order to illustrate some of the abstract concepts described above, we conclude this Introduction with a concrete example, pure (non-abelian) Yang-Mills theory.
Consider a trivial principal $G$-bundle, $Q$, over a 4-dimensional space-time, say Minkowski space $\bbM^4$, and, as a space of fields, $\mathcal A$, the connections thereon. Since $Q$ is assumed to be trivial, these are globally defined one-forms on $\bbM^4$ with values in the Lie algebra, $\g$, of $G$:
$$
\mathcal A:=\Omega^1(\bbM^4,\g)\,.
$$
On $\mathcal A$ we define the pure Yang-Mills action
$$
S_0(A):=\Int{\bbM^4}\tr (F(A)\wedge *F(A))\,,\quad F(A):=dA+\frac{1}{2}[A\wedge A]\,,
$$
where $*$ denotes the Hodge-dual w.r.t. the Minkowski metric.
The infinitesimal gauge symmetries, 
$$
A\mapsto A+\vc X_\epsilon(A)\,,
$$
are given by vector fields on $\mathcal A$ of the form
\bb{IntroExSymm}
\vc X_\epsilon(A)=D_A\epsilon=d\epsilon+[A,\epsilon]\,,\quad \epsilon\in\Omega^0(\bbM^4,\g)\,.
\ee
Corresponding to these symmetries we introduce {\em ghost fields}
$$
\beta\in\Omega^0(\bbM^4,\g[1])\,,
$$
where the integer in brackets denotes the {\em ghost degree}. (Fields of even/odd ghost degree are called {\em bosonic/fermionic}, respectively.)
Furthermore, we introduce {\em antighosts} (not to be confused with the anti-fields of the ghosts!) and {\em Lagrange multipliers}
$$
\bar\beta\in \Omega^4(\bbM^4,\g^*[-1])\,,\quad \lambda \in \Omega^4(\bbM^4,\g^*)\,.
$$
The fields, ghosts, anti-ghosts and Lagrange multipliers determine an infinite-dimensional graded manifold, $\M$. Functions on $\M$ are local functionals of sections of a trivial graded vector bundle over $\bbM^4$.
The symmetry (\ref{IntroExSymm}) gives rise to the $BRST$-operator
\begin{align*}
\mathcal X(A)&=D_A\beta\,,\\
\mathcal X(\beta)&=-\frac{1}{2}[\beta,\beta]\,,
\end{align*}
which is a co-boundary operator as the Lie bracket satisfies the Leibniz and Jacobi identities.
On the anti-ghosts and Lagrange multipliers we define
\begin{align*}
\mathcal X(\bar\beta)&=\lambda\,,\\
\mathcal X(\lambda)&=0\,.
\end{align*}
Finally, every field in $\M$ is paired with an {\em anti-field}, i.e., we introduce
\begin{align*}
A^\dag&\in \Omega^3(\bbM^4,\g^*[-1])\,,\\
\beta^\dag&\in \Omega^4(\bbM^4,\g^*[-2])\,,\\
\bar\beta^\dag&\in \Omega^0(\bbM^4,\g)\,,\\
\lambda^\dag&\in \Omega^0(\bbM^4,\g[-1])\,.
\end{align*}
The fields and anti-fields together form the odd symplectic space $\E=\Pi T^*\M$; the gradings are chosen such that its natural odd symplectic structure has ghost degree $-1$.

Pairing each anti-field with the $BRST$-transformation of the corresponding field according to (\ref{IntroS1}), we extend the action $S_0$ to
$$
S=S_0+\Int{\bbM^4}\langle A^\dag\wedge D_A\beta\rangle-\langle \beta^\dag,\frac{1}{2}[\beta,\beta]\rangle+\langle\bar\beta^\dag,\lambda\rangle\,,
$$
where $\langle.,.\rangle$ denotes the natural pairing of elements in $\g$ with elements in $\g^*$.
The $CME$ $\{S,S\}=0$ follows from the facts that $\mathcal X$ squares to zero and $S_0$ is gauge-invariant.
Furthermore, gauge-invariant functionals on $\mathcal A$ have vanishing odd Poisson bracket with $S$, so that condition (\ref{Introcond1}) is satisfied.

The $BV$ Laplacian $\Delta_\Omega$ appearing in the $QME$ is given by the operator
$$
\Delta_\Omega=\Int{\bbM^4}\left\lbrace\frac{\delta^2}{\delta A\delta A^\dag}+\frac{\delta^2}{\delta \beta\delta \beta^\dag}+\frac{\delta^2}{\delta \bar\beta\delta \bar\beta^\dag}+\frac{\delta^2}{\delta \lambda\delta \lambda^\dag}\right\rbrace\,,
$$
and we find that $\Delta_\Omega S=0$. Hence, $S$ is a solution of the $QME$.
Notice, however, that $\Delta_\Omega$ is singular.
We refer the reader to Costello \cite{Costello} for a {\em renormalized} version of the $QME$ and its solutions for the example considered here.

Next, we impose a general gauge-fixing condition
\bb{Introgf}
G(A)=0\,,
\ee
where $G:\,\Omega^1(\bbM^4,\g)\rightarrow\Omega^0(\bbM^4,\g)$ is a local map, i.e., $(G(A))(x)$ depends only on a finite number of derivatives of $A$ at $x$.
The Lorentz-gauge $G(A)=\partial^\mu A_\mu$ provides an example.
Condition (\ref{Introgf}) is implemented via the gauge-fixing fermion
$$
\Psi:=\Int{\bbM^4}\langle\bar\beta,G(A)\rangle+i\Int{\bbM^4}\alpha(\bar\beta,*\lambda)\,,\quad\alpha\,:\,\text{metric on $\g^*$}\,,
$$
where the second term is added in order to achieve a Gaussian average around the gauge-fixing (\ref{Introgf}).
The associated Lagrangian submanifold, $\L_\Psi$, is determined by the equations
$$
A^\dag=\frac{\delta\Psi}{\delta A}=\delta_A G^*(\bar\beta)\,,\quad \bar\beta^\dag=\frac{\delta\Psi}{\delta \bar\beta}=G(A)+i\alpha(*\lambda)\,,\quad \beta^\dag=\lambda^\dag=0\,,
$$
where $\delta_A G^*\,:\,\Omega^4(\bbM^4,\g^*)\rightarrow \Omega^3(\bbM^4,\g^*)$ is the dual of the derivative, $\delta_AG\,:\,\Omega^1(\bbM^4,\g)\rightarrow \Omega^0(\bbM^4,\g)$, of $G$ at $A$.
The restriction of the action $S$ to $\L_\Psi$ reads
\bb{IntroYMgf}
S_{\arrowvert \L_\Psi}=S_0+\Int{\bbM^4}\langle\bar\beta,\delta_A G(D_A\beta)\rangle+\langle \lambda,G(A)\rangle+i\alpha(\lambda,*\lambda)\,.
\ee
The first term under the integral yields the usual Faddeev-Popov determinant after Berezin integration, while the last two terms implement the gauge-fixing.

Note that if $\alpha=0$ we actually do not need to introduce Lagrange multipliers and antighosts in order to implement the gauge-fixing (\ref{Introgf}). Instead, we can use the Lagrangian submanifold, $\L$, defined by the equations
$$
G(A)=0\,,\quad \beta^\dag=0\,,\quad A^\dag=\delta_A G^*(\bar\beta)\quad\text{for}\quad \bar\beta\in \Omega^4(\bbM^4,\g^*[-1])\,.
$$
%In order for $\L$ to be Lagrangian, we need to integrate only over anti-fields of the form $A^\dag=\delta_A G^*(\bar\beta)$, where $\bar\beta\in \Omega^4(\bbM^4,\g^*[-1])$.
%Hence, integrating $\exp(iS)$ over $\L$, without the last term containing the Lagrange multipliers, is equivalent to integrating $\exp(iS)$ over $\L_\Psi$.

The organization of this paper is as follows:
In Section 2 we first introduce general notions of graded differential geometry.
Then we review homological perturbation theory and use it to define semidensities and develop their integration theory.
In Section 3 we construct an action satisfying the $CME$ (\ref{IntroCME}) and property (\ref{Introcond1}), for the finite dimensional Lagrangian systems introduced above.
Homological Perturbation Theory will be the essential tool for the construction of a solution of the $CME$ for an open symmetry.
We discuss the role of the operator $\Delta_\Omega$ in connection with localization and illustrate our formalism with Duistermaat-Heckman localization. 
We conclude with a comment on anomalies, i.e., on obstructions to solving the $QME$.  

\subsection*{Acknowledgements}
The first author is indebted to Pavol Severa and Giovanni Felder for enlightening discussions about the $BV$ formalism.
Furthermore, he gratefully acknowledges financial support by the Swiss National Foundation, the ETH Zurich and the University of Geneva.
The second author also gratefully acknowledges support by ETH Zurich.
\section{Mathematical Formalism}
\subsection{Elements of Graded Algebra}
The notion of a \emph{grading} for an algebraic object is defined over any monoid $G$. For us, $G = \N, \Z$ or $\Z_2= \{ 0, 1 \}$. Instead of using the term $\bbZ_2$-graded we also use the prefix {\em "super"}.

A {\emph{$G$--graded ring}} is a ring, $R = \bigoplus_{i \in G} R_i$, such that $R_iR_j \subseteq R_{ij}$, where $ij$ is the product of $i, j \in G$. 
Given a $G$-graded ring $R$, a {\emph{$G$--graded $R$--module}} is an $R$--module $M = \bigoplus_{i \in G} M_i$ such that $R_iM_j \subseteq M_{ij}$. A {\emph{$G$--graded $R$--algebra}} is an $R$-algebra $A = \bigoplus_{i \in G} A_i$ such that $A_iA_j \subseteq A_{ij}$ and $R_iA_j \subseteq A_{ij}$.
An element $x \in M_i$ (or $A_i$ or $R_i$) is called \emph{homogeneous of degree $i$}, and we write $|x| = i$. 

Let $M$, $N$ be graded $R$-modules. Then we define
\bb{DGHom}
{\rm{Hom}}(M,N)^j = \{ \varphi: M \rightarrow N \ | \ \varphi \ {\text{is  $R$-linear  and}} \ \varphi(M_i) \subseteq N_{i+j} \ \},
\ee
and 
$$
{\rm{Hom}}(M,N) = \bigoplus_{j \in G} {\rm{Hom}}(M,N)^j\,.
$$

%For $G = \Z$ or $\Z_2$, there is a functor, $\Pi$, from the category of $G$--graded modules to itself, called the \emph{shift} or \emph{suspension}, defined by
%
%\begin{equation}\label{shift}
%(\Pi M)^i = M^{i-1}.
%\end{equation}
%
%In the literature $\Pi M$ is also denoted by $M[-1]$. 

%For example, the \emph{Grassman algebra} $A = k[\xi_1, \ldots, \xi_n]$ over the field $k$ of characteristic $r \neq 2$, generated by $\xi_1, \ldots, \xi_n$ subject to the relations $\xi_i \xi_j = - \xi_j \xi_i$, for $1 \leq i, j \leq n$, is a super--commutative algebra with $|\xi_i| = 1$. 

The \emph{tensor algebra} of $M$ over $R$, 
\begin{equation}\label{tensalg}
T_R(M) = \bigoplus_{n=0}^{\infty} M^{\otimes n},
\end{equation}
where $M^0 = R$ and the tensor product is taken w.r.t. $R$, is a bi--graded associative algebra with multiplication $M^{\otimes n} \otimes M^{\otimes m} \rightarrow M^{\otimes (n+m)}$ given by $(m_1, m_2) \mapsto m_1 \otimes m_2$.
Let $J_S$ be the ideal of $T_R(M)$ generated by {\em graded commutators},
$$
[[m_1, m_2]] = m_1 \otimes m_2 - (-1)^{|m_1||m_2|}m_2 \otimes m_1\,,
$$
with $m_1, m_2$ homogeneous in $M$. Then
\bb{DGsymmalg}
S_R(M) = \bigoplus_{n=0}^{\infty}S^n_R(M) = T_R(M) / J_S
\ee
is the {\emph{graded commutative algebra}} generated by $M$. 
Similarly, the {\em graded anti-commutative algebra} generated by $M$ is given by 
\bb{DGasymmalg}
{\bigwedge}_R(M) = T_R(M) / J_{\bigwedge}\,,
\ee 
where $J_{\bigwedge}$ is the ideal of $T_R(M)$ generated by $m_1 \otimes m_2 + (-1)^{|m_1||m_2|}m_2 \otimes m_1$, with $m_1, m_2$ homogeneous in $M$.

\subsection{Graded Manifolds}
Consider a bundle, $E$, of graded vector spaces over a smooth manifold $M$.
By this we mean that the local sections, $\Gamma(U,E)$, for any local neighbourhood $U\subseteq M$, are $C^\infty(M)$-modules freely generated by homogeneous sections.
To each open subset, $U\subseteq M$, we assign the graded commutative $C^\infty(U)$-algebra
\bb{gradedmfd}
C^\infty(\U):=S_{C^\infty(U)}(\Gamma(U,E^*))=\Gamma(U,S_\bbR(E^*))\,
\ee
as defined in (\ref{DGsymmalg}). In (\ref{gradedmfd}), $E^*$ denotes the bundle dual to $E$, and the degree of the dual of a homogeneous section is the negative of the degree of the original section.
Just as the assignment $U\mapsto C^\infty(U)$ may be viewed as defining a smooth structure on a topological manifold, we shall use the assignment $U\mapsto C^\infty(\U)$ to equip $M$ with the structure of a graded manifold. 
We denote the latter by $\M$ and write $C^\infty(\M)$ if $U=M$ in (\ref{gradedmfd}).
Elements of $C^\infty(\M)$ are called {\em graded functions}, their degree is called {\em ghost degree}.

This subsection is devoted to an abstract definition of a graded manifold inspired by the properties of the assignment (\ref{gradedmfd}). We will see that the isomorphism classes of graded manifolds are in bijection with the isomorphism classes of graded vector bundles. However, the category of graded vector bundles has fewer morphisms than the category of graded manifolds, since morphisms of graded manifolds do not need to be linear in the generators of the fibres of the associated graded vector bundle.

That the assignment (\ref{gradedmfd}) behaves nicely under restrictions to smaller open sets is encoded in the definition of a {\em sheaf}.
In general, a sheaf assigns to every open subset, $U$, of a topological space $M$ some algebraic structure $A(U)$ (in (\ref{gradedmfd}), $A(U)$ is a graded commutative $C^\infty(U)$-algebra), such that there is a restriction map
$$
r^U_V:A(U)\longrightarrow A(V)\,,\quad V\subseteq U\subseteq M\,,
$$
satisfying the following natural axioms:
For open subsets $W\subseteq V\subseteq U\subseteq M$:
\begin{itemize}
\item[(S1)] $r^U_W = r^V_W \circ r^U_V$; 
\item[(S2)] $r^U_U = \id_{A(U)}$.
\end{itemize}
Furthermore, given any collection $U = \bigcup_i U_i$ of open subsets in $M$,
\begin{itemize}
\item[(S3)] for $s, t \in A(U)$, the conditions $r^U_{U_i} (s) = r^U_{U_i} (t)$, for all $i$, imply $s = t$; 
\item[(S4)] for $s_i \in A(U_i)$, the conditions $r^{U_i}_{U_i \cap U_j} (s_i) = r^{U_j}_{U_i \cap U_j} (s_j)$, for all $i, j$, imply that there is an $s \in A(U)$ such that $r^U_{U_i} (s) = s_i$, for all $i$.
\end{itemize}
The sheaf (\ref{gradedmfd}) has an additional property, which must be satisfied by any graded manifold:
it is {\em locally trivial}, which means that for any local neighbourhood (coordinate patch) $U\subseteq M$ with local coordinates $\{x^1,\dots,x^{(m-m')}\}$ there are graded functions
\begin{equation}\label{localcoord}
\{y^1\ldots, y^{m'}, \beta^1, \ldots, \beta^n\}, 
\end{equation}
such that $C^\infty(\U)$ is generated - as a $C^\infty(U)$-algebra - by the even functions $\{y^1, \ldots, y^{m'}\}$ and the odd functions $\{\beta^1, \ldots, \beta^n\}$ and such that the requirements of graded commutativity yield the only relations among them.
We say that $C^\infty(\U)$ is freely generated by $m'$ even and $n$ odd generators.
We then say that the graded manifold $\M$ over a manifold $M$ of dimension $(m-m')$ has dimension $(m,n)$ and call $\{x^1,\dots,x^{(m-m')},y^1\ldots, y^{m'}, \beta^1, \ldots, \beta^n\}$ a set of {\em local coordinates} on $\U$. 
Hence, in local coordinates, any $f \in C^{\infty}(\U)$ can be written as
\begin{equation*}
f =  \sum_{i_1,\dots,i_{m'},\alpha_1, \ldots ,\alpha_n} f_{i_1,\dots,i_{m'},\alpha_1 \ldots \alpha_n}(x) (y^1)^{i_1} \ldots (y^{m'})^{i_{m'}} (\beta^1)^{\alpha_1} \ldots (\beta^n)^{\alpha_n}\,,
\end{equation*}
where the $i$'s are in $\mathbb N_0$, the $\alpha$'s in $\mathbb Z_2$ and $f_{i_1,\dots,i_{m'},\alpha_1, \ldots, \alpha_n}\in C^\infty(U)$. 
\begin{definition}
A {\em graded manifold} of dimension $(m,n)$ over a smooth, $(m-m')$-dimensional manifold $M$ is a sheaf of graded commutative $C^\infty(U)$-algebras
$$
M\supseteq U\mapsto C^\infty(\U)\,
$$
that are locally freely generated by $m'$ even and $n$ odd generators.
\end{definition}
A {\em morphism} $(\phi, \phi^*): (M, C^\infty(\M)) \rightarrow (N, C^\infty(\mathcal N))$ of graded algebras consists of a smooth map $\phi: M \rightarrow N$ and, for each open set $U\subseteq N$ and $V:=\phi^{-1}(U)$, a morphism of graded commutative algebras $\phi^*:C^\infty(\U)\rightarrow C^\infty(\mathcal V)$, i.e.,
$$
\phi^*(f):=f\circ \phi\quad\text{and}\quad\phi^*(fab)=\phi^*(f)\phi^*(a)\phi^*(b)\,,\quad\text{with}\quad |\phi^*(a)|=|a|\,,\,|\phi^*(b)|=|b|\,,
$$
for all $f\in C^\infty(U)$ and all homogeneous $a,b\in C^\infty(\U)$.

If the grading under consideration is a $ \bbZ_2$-grading, graded manifolds are called {\em super-manifolds}.
Batchelor \cite{Ba79} showed that, up to isomorphism, every super-manifold arises from a unique graded vector bundle as in (\ref{gradedmfd}).
Her arguments extend to general gradings, and hence, up to isomorphism, every graded manifold arises from a unique graded vector bundle, which we call the {\em graded vector bundle associated with the graded manifold}.
However, since the linear structure of the latter is not remembered by the graded manifold, not every morphism of a graded manifold arises from a morphism of the associated graded vector bundle.

Note that $C^\infty(\M)$ has a distinguished ideal, $I_0$, generated by its odd functions (it is clear that this notion is independent of the choice of coordinates). The ideal $I_0$ defines a distinguished even submanifold $\iota:\, M_0\hookrightarrow\M$, called the {\em body} and defined by the natural projection
$$
\iota^*:C^\infty(\M)\longrightarrow C^\infty(M_0)=C^\infty(\M)/I_0\,.
$$
While there is no canonical projection $\pi:\M\rightarrow M_0$, the existence of a graded vector bundle associated with $\mathcal M$ guarantees the existence of a projection: The body $M_0$ is the graded manifold arising from the even subbundle of the vector bundle associated with $\mathcal M$, and the projection of this bundle onto its even subbundle yields a projection  $\pi: \mathcal M \rightarrow M_0$.
\subsection{Graded Differential Geometry}
In this subsection, we introduce graded versions of (multi-) vector fields and differential forms, together with the operations of exterior differentiation, interior multiplication and Lie derivation and establish the commutation relations among these operations.

We define a homogeneous {\em (left-) vector field} of ghost degree $q$, denoted by $\lder{X}\in \mathcal X^{1,q}(\M)$, as a linear map of degree $q$, $\lder{X}\in {\rm{Hom}}(C^\infty(\M),C^\infty(\M))^q$ (see eq. (\ref{DGHom})), satisfying the {\em graded Leibniz rule}
\bb{DGgLeibnizvf}
\overset{\rightarrow}{\vc X}(f\cdot g)=\overset{\rightarrow}{\vc X}(f)\cdot g+(-1)^{(|\lder{X}|+1)|f|}f\cdot\overset{\rightarrow}{\vc X}(g)\,,
\ee
where we define the {\em total degree} of the vector field $\lder{X}$ to be
$$
|\lder{X}|:=q-1\,.
$$
The graded commutator,
\bb{DGgLA}
[[\overset{\rightarrow}{\vc X},\overset{\rightarrow}{\vc Y}]] = \overset{\rightarrow}{\vc X}\overset{\rightarrow}{\vc Y}-(-1)^{(|\lder{X}|+1)(|\lder{Y}|+1)}\overset{\rightarrow}{\vc Y}\overset{\rightarrow}{\vc X}\,,
\ee
equips $\mathcal X^{1,\bullet}(\M)$ with the structure of a graded Lie algebra, i.e., it satisfies the {\em graded Jacobi identity}
\bb{DGJacobi}
(-1)^{(|\lder{X}|+1)(|\lder{Z}|+1)}[[\lder{X},[[\lder{Y},\lder{Z}]]]]+\quad\text{cyclic permutations}=0\,.
\ee

Sometimes matters simplify by using {\em right vector fields}, which are defined as follows:
\bb{DGrightvf}
(f)\overset{\leftarrow}{\vc X}:=(-1)^{(|\lder{X}|+1)(|f|+1)}\overset{\rightarrow}{\vc X}(f)\,.
\ee
From (\ref{DGgLeibnizvf}) and (\ref{DGrightvf}) we deduce the Leibniz rule for right vector fields
\bb{BVDG0'}
(f\cdot g)\overset{\leftarrow}{\vc X}=f\cdot (g)\overset{\leftarrow}{\vc X}+(-1)^{(|\vc X|+1)|g|}(f)\overset{\leftarrow}{\vc X}\cdot g\,.
\ee

In local coordinates $\{z^1, \ldots, z^{m+n}\}$, the left vector fields \\
$\{ \frac{\lpartial}{\partial z^i}\}_{i = 1, \ldots, m+n}$, determined by
\begin{equation}\label{basisTU}
\frac{\overset{\rightarrow}{\partial}}{\partial z^i}(z^j) = \delta^j_i\,,
\end{equation}
form a local basis of $\mathcal X^{1,\bullet}(\U)$ as a $C^\infty(\U)$-module.
We then obtain
\begin{equation*}
(z^j)\frac{\overset{\leftarrow}{\partial}}{\partial z^i} = \delta^j_i\,,
\end{equation*}
for the associated right vector fields, and the graded commutation relations
\begin{equation}\label{gradedcomder}
\left[\left[\frac{\overset{\rightarrow}{\partial}}{\partial z^i},\frac{\overset{\rightarrow}
{\partial}}{\partial z^j}\right]\right] = \left[\left[\frac{\overset{\leftarrow}{\partial}}{\partial z^i},\frac{\overset{\leftarrow}{\partial}}{\partial z^j}\right]\right] = 0 
\end{equation}
hold.
By definition,
\begin{equation*}
\left |\frac{\overset{\leftarrow}{\partial}}{\partial z^i}\right |  = \left |\frac{\overset{\rightarrow}{\partial}}
{\partial z^i} \right | = - |z_i|-1.
\end{equation*}

Next, we define (see (\ref{DGsymmalg}))
$$
\mathcal X^{p,\bullet}(\M):=S^p_{C^\infty(\M)}(\mathcal X^{1,\bullet}(\M))\,,
$$
and the graded commutative algebra
$$
\mathcal X(\M):=\bigoplus_{p,q}\mathcal X^{p,q}(\M)\,,
$$
where "graded commutative" is to be understood w.r.t. the total degree:
An element, $\chi$, of $\mathcal X^{p,q}(\M)$ is called a {\em multi-vector field} and has total degree $|\chi|=q-p$.
We use the symbol $\wedge$ for the associative product in $\mathcal X^{p,q}(\M)$, i.e., for $\chi_1\in\mathcal X^{p_1,q_1}(\M)$ and $\chi_2\in\mathcal X^{p_2,q_2}(\M)$ we have
\bb{DGvectorwedge}
\chi_1\wedge\chi_2=(-1)^{|\chi_1||\chi_2|}\chi_2\wedge\chi_1\in\mathcal X^{p_1+p_2,q_1+q_2}(\M)\,.
\ee

We introduce {\em graded differential forms} starting with one-forms, which are defined to be dual to the (left-) vector fields:
\bb{DGoneform}
\Omega^{1,q}(\M):={\rm Hom}(\mathcal X^{1,\bullet}(\M),C^\infty(\M))^q\,.
\ee
Differential forms of higher form degree are elements of the graded commutative algebra
$$
\Omega(\M):=\bigoplus\limits_{p,q}\Omega^{p,q}(\M)\,,
$$
with
$$
\Omega^{p,\bullet}(\M):=S^p_{C^\infty(\M)}(\Omega^{1,\bullet}(\M))\,.
$$
Again, "graded commutative" is understood w.r.t. the total degree, $|\eta|:=p+q$, for $\eta\in\Omega^{p,q}(\M)$.
The associative product in $\Omega^{p,q}(\M)$ is denoted by $\wedge$, and
\begin{equation}\label{gradedcomforms}
\eta_1 \wedge \eta_2 = (-1)^{|\eta_1| |\eta_2|} \eta_2 \wedge \eta_1\,.
\end{equation}

According to (\ref{DGoneform}), vector fields pair with one-forms. We write
$$
\iota_{\lder{X}}\eta:=\eta(\lder{X})
$$
and extend $\iota_{\lder{X}}$ to an {\em interior multiplication} imposing the graded Leibniz rule
\[ \iota_{\lder{X}}(\eta_1 \wedge \eta_2) = \iota_{\lder{X}} \eta_1 \wedge \eta_2 + (-1)^{|\eta_1| |\lder{X}|} \eta_1 \wedge \iota_{\lder{X}} \eta_2,\]
for $\eta_1, \eta_2 \in \Omega(\M)$, as well as the rule
\bb{DGIntMult2}
\iota_{\lder{X}\wedge\lder{Y}}:=\iota_{\lder{X}}\iota_{\lder{Y}}\,,
\ee
for $\lder{X}, \lder{Y}\in\mathcal X^{1,\bullet}(\M)$.
For $\chi\in\mathcal X^{p,q}(\M)$, the operator $\iota_{\chi}$ has degree 
$$
|\iota_{\chi}|=|\chi|=q-p\,,
$$
and it follows immediately from (\ref{DGvectorwedge}) and (\ref{DGIntMult2}) that
$$
[[\iota_{\lder{X}},\iota_{\lder{Y}}]]=0\,,\quad \lder{X},\lder{Y}\in \mathcal X^{1,\bullet}(\M)\,.
$$

Next, we show that there is a unique {\em exterior differential}, $d$, of degree $|d|=1$, satisfying the following equalities:
\begin{align}
d(\eta_1\wedge\eta_2)&=d\eta_1\wedge \eta_2+(-1)^{|\eta_1|}\eta_1\wedge d\eta_2\quad\text{(graded Leibniz rule)}\,,\label{DGdLeibniz}\\
\iota_{\lder{X}}df&=\lder{X}(f)\,,\label{DGdf}\\
\frac{1}{2}[[d,d]]&=d^2=0\,.\label{DGdsquare}
\end{align}
From (\ref{DGdLeibniz}) to (\ref{DGdsquare}) it is clear that we only need to verify condition (\ref{DGdsquare}) on functions.
We introduce local coordinates $\{z^1, \ldots, z^{m+n}\}$ and calculate, using (\ref{DGdf}),
$$
df=dz^i\frac{\lpartial f}{\partial z^i}\,.
$$
Hence,
$$
d^2f=(-1)^{|z^i|+1}dz^i\wedge dz^j\frac{\lpartial^2 f}{\partial z^j\partial z^i}\,,
$$
which vanishes due to (\ref{DGgLA}), (\ref{gradedcomder}) and (\ref{gradedcomforms}).

Finally, the \emph{Lie--derivative} of a differential form, $\eta$, along a vector field, $\lder{X}$, is given by
\begin{equation}\label{BVDefLie}
L_{\lder{X}}\eta := [[\iota_{\lder{X}}, d ]] \eta = \iota_{\lder{X}}(d \eta) - (-1)^{|\lder{X}|} d(\iota_{\lder{X}} \eta)\,.
\end{equation}
It has degree $|L_{\lder {X}}|=|\lder{X}|+1$ and obeys the commutation relations
\begin{align}
[[L_{\lder{X}},L_{\lder{Y}}]]&=L_{[[\lder{X},\lder{Y}]]}\,,\label{DGLL}\\
[[L_{\lder{X}},\iota_{\lder{Y}}]]&=\iota_{[[\lder{X},\lder{Y}]]}\,,\label{DGLiota}
\end{align}
with $[[\lder{X},\lder{Y}]]$ as defined in (\ref{DGgLA}).
\subsection{Odd Symplectic Manifolds}
An \emph{odd symplectic manifold, $(\mathcal E, \omega)$}, is a graded manifold $ \E$ equipped with an {\em odd symplectic structure}, more precisely, with a closed, non-degenerate two-form $\omega$ of ghost degree $-1$, i.e., $\omega\in\Omega^{2,-1}( \E)$.

To each function $f\in C^\infty( \E)$ we associate a {\em Hamiltonian vector field}, $\lder{X}_f$, defined by
\bb{DGHamVF}
\iota_{\lder{X}_f}\omega:=(-1)^{|f|}df\,,
\ee
of degree
\bb{DGdegHamVF}
|\lder{X}_f|=|f|\,.
\ee
The odd symplectic form $\omega$ yields the \emph{odd Poisson structure} defined by
\bb{DGPoisson}
\{f,g\}:=\lder{X}_f(g)=\iota_{\lder{X}_f}dg=(-1)^{|g|}\iota_{\lder{X}_f\wedge\lder{X}_g}\omega\,,\quad f,g\in C^\infty( \E)\,,
\ee
where the second equality follows from (\ref{DGdf}) and the third from (\ref{DGHamVF}) and (\ref{DGIntMult2}).
From (\ref{DGLiota}), (\ref{DGHamVF}) and (\ref{DGPoisson}) we derive the formula
\bb{DGPoissonHamVF}
\lder{X}_{\{f,g\}}=[[\lder{X}_f,\lder{X}_g]]\,.
\ee

From these formulae together with (\ref{DGgLeibnizvf}) to (\ref{DGJacobi}) it is straightforward to derive graded versions of the usual properties of a Poisson bracket.
\begin{lemma}\label{BVLPoisson}
For an odd symplectic graded manifold $( \E,\omega)$, the bracket defined in (\ref{DGPoisson}) satisfies the following properties, for all homogeneous $f, g, h\in C^\infty( \E)$:
\begin{enumerate}
\item Graded Commutation Relation: $\{f,g\}=-(-1)^{(|f|+1)(|g|+1)}\{g,f\}$, \label{PBProp1} \\
\item Graded Jacobi Identity: $(-1)^{(|h|+1)(|g|+1)}\{h,\{f,g\}\}+\text{cyclic permutations}=0$,  \quad  \label{Jacobi}\\
\item Graded Leibniz Rule: $\{f,gh\}=\{f,g\}h+(-1)^{|g|(|f|+1)}g\{f,h\}$. \label{BVLeibniz}
\end{enumerate}
\end{lemma}
In the mathematics literature, a graded commutative associative algebra together with a bracket satisfying the equations in Lemma \ref{BVLPoisson} is called a {\em Gerstenhaber algebra}.

As we have seen in Subsection 2.2, the graded manifold $ \E$ has a distinguished even submanifold, namely its body, $E_0$.
Hence, $\E$ is associated with some odd vector bundle, $E$, over $E_0$.
Since $\omega$ is of ghost degree $-1$ it induces a symplectic structure on $E$, such that the fibres and the zero section of $E$ are Lagrangian.
It is then a standard result from symplectic geometry that $E$ is isomorphic to the cotangent bundle over $E_0$.
%A section $p \in \Gamma(E^\ast)$ gives rise to a graded function in $C^\infty(\mathcal E)$ (see (\ref{gradedmfd})), which we also denote by $p$, and to which we assign the derivation
%\bb{BVDarboux}
%\iota^\ast\{\pi^\ast(\cdot),p\}: C^\infty(E_0) \rightarrow C^\infty(E_0).
%\ee
%Here $\iota: E_0 \rightarrow \mathcal E$ is the natural inclusion of the body and $\pi: \mathcal E \rightarrow E_0$ is any projection onto the body.
%Since $\omega$ is non-degenerate, each derivation on $C^\infty(E_0)$ can be written in the form (\ref{BVDarboux}).
Since $\omega$ has ghost degree $-1$, we need to increase the degree of the fibres of this co-tangent bundle by one. We denote this by $E=\Pi T^*E_0$.
We have proven the following version of the Darboux theorem for graded odd symplectic manifolds:
\begin{lemma}\label{DGDarboux}
Any graded manifold $ \E$, equipped with a non-degenerate closed two-form $\omega$ of ghost degree $-1$, is associated with the graded vector bundle $\Pi T^*E_0$ over its body, where the functor $\Pi$ increases the ghost degree by one.
\end{lemma}
Choose a projection $\pi:\, \E\rightarrow E_0$ and a set of local\footnote{Remember that $E_0$ is an even manifold associated with some even vector bundle over a manifold $M$. Locality is to be understood w.r.t. $M$.} 
coordinates, $\{\hat x^i\}$, on $E_0$.
Together with the corresponding local coordinate vector fields $\{\partial/\partial\hat x^i\}$, they define a set of local graded coordinate functions, $\{x^i\,,\,x^\dag_i\}$ (with $|x^\dag_i|=-|x^i|-1$), on $\E$ given by the equations
\bb{Darbouxproof}
x^i:=\pi^*(\hat x^i)\,,\quad \{\pi^*(.),x^\dag_i\}:=\pi^*\left(\frac{\partial}{\partial\hat x^i}(.)\right)\,.
\ee
We call $x^\dagger_i$ the \emph{anti-coordinate} corresponding to the coordinate $x^i$, and the coordinates together with their anti-coordinates form  a system of local \emph{Darboux coordinates}.  With respect to these Darboux coordinates, the odd symplectic two-form $\omega$ is given by
\begin{equation}\label{omindarb}
\omega = dx^i \wedge dx_i^{\dagger}\,.
\end{equation}
In order to see this, we derive from (\ref{DGrightvf}), (\ref{DGHamVF}), (\ref{DGPoisson}) and (\ref{omindarb}) the local expression for the odd Poisson structure
\begin{equation*}
\{f,g\} = \frac{f \overset{\leftarrow}\partial}{\partial x^i}\frac{\overset{\rightarrow}\partial g}{\partial x_i^{\dagger}} - \frac{f \overset{\leftarrow}\partial}{\partial x_i^{\dagger}}\frac{\overset{\rightarrow}\partial g}{\partial x^i}, \quad {\rm for} \,f, g \,\in \,C^{\infty}( \E)\,.
\end{equation*}
For $f=x^i$, $g=x^\dag_j$, we find agreement with (\ref{Darbouxproof}).

The identification of graded functions on $\mathcal E$ with multi-vector fields on $E_0$ depends on the choice of projection $\pi: \mathcal E \rightarrow E_0$, and, by abuse of notation, we write $\Pi T^\ast E_0$ for $\mathcal E$ together with a choice of projection $\pi$. 
Analogously, we define the odd cotangent bundle, $\Pi T^\ast \mathcal M$, of an arbitrary graded manifold, $\mathcal M$.
If $\{z^i\}$ are local coordinates on $\M$, corresponding anti-coordinates, $\{z^\dag_i\}$, have degree $|z^\dag_i|=-|z^i|-1$ and the two-form $\omega$, which, in local coordinates, is given by
$$
\omega=dz^i\wedge dz^\dag_i\,
$$
defines a natural odd symplectic structure on $\E$.
The corresponding odd Poisson structure is then given by
$$
\{f,g\} = (-1)^{|z^i|}\left(\frac{f \overset{\leftarrow}\partial}{\partial z^i}\frac{\overset{\rightarrow}\partial g}{\partial z_i^{\dagger}} - \frac{f \overset{\leftarrow}\partial}{\partial z_i^{\dagger}}\frac{\overset{\rightarrow}\partial g}{\partial z^i}\right), \quad {\rm for} \,f, g \,\in \,C^{\infty}( \E).
$$

Obviously, odd symplectic manifolds are of dimension $(m,m)$, and, being odd, the symplectic structure vanishes on its $(m,0)$-dimensional body. Hence, the body of an odd symplectic manifold is a {\em Lagrangian submanifold}, locally given by the ideal
$$
I_0:=\langle x^\dag_i\rangle_{i=1,\dots,m}\,,
$$
for any choice of local Darboux coordinates.
In the $BV$-framework, we consider more general $(k,m-k)$-dimensional Lagrangian submanifolds given by a sheaf of ideals,
\bb{DGLagSM}
I(U)=\langle f^i,\zeta_\alpha\rangle^{i=k+1,\dots,m}_{\alpha=1,\dots,k}\,,
\ee
where the $f$'s ($\zeta$'s) are homogeneous even (odd) functions in $C^\infty(\U)$ (also called {\em constraints}) that are in involution, i.e.,
\bb{DGLagSM2}
\{f^i,f^j\}=\{f^i,\zeta_\alpha\}=\{\zeta_\alpha,\zeta_\beta\}=0\,,\quad k+1\leq i,j\leq m\,,\quad 1\leq \alpha,\beta\leq k\,.
\ee

Consider a Lagrangian submanifold, $\L$, of $ \E$, given by the sheaf of ideals (\ref{DGLagSM}).
The body, $L\subseteq E_0$, of $\L$ is defined as
$$
L:=\L\cap E_0\,,
$$
i.e., $C^\infty(L):=C^\infty(\E)/(I\cup I_0)$. 
A projection, $\pi:  \E \rightarrow E_0$, of $ \E$ onto its body is said to be {\em adapted} to $\L$, write $\pi=\pi_\L$, if
$$
\pi_\L(\L)=L\,,
$$
i.e., if $\pi_\L^*(f+I_0)\in I$, for all $f\in I$.

\begin{lemma}\label{BVLadapted}
For any Lagrangian submanifold $\L$ of the odd symplectic manifold $ \E$ there is a projection $\pi_\L:  \E \rightarrow E_0$ adapted to $\L$.
\end{lemma}
\begin{proof}
Let $\L$ be defined by the constraints (\ref{DGLagSM}).
Since the Hamiltonian vector fields $\lder {X}_{f^i}$, $\lder{X}_{\zeta_\alpha}$ are in involution because of (\ref{DGPoissonHamVF}) and (\ref{DGLagSM2}), there are local Darboux coordinates\\ $\{x^1, \ldots, x^m, x_1^{\dagger}, \ldots, x_m^{\dagger}\}$ such that
$$
x^i:=f^i\,,\quad k+1\leq i\leq m\,,\quad x^\dag_\alpha:=\zeta_\alpha\,,\quad 1\leq \alpha\leq k
$$
and 
$$
\frac{\lpartial}{\partial x^\alpha}=\lder {X}_{\zeta_\alpha}\,,\quad \quad 1\leq \alpha\leq k\,,\quad \frac{\lpartial}{\partial x^\dag_i}=\lder {X}_{f^i}\,, \quad k+1\leq i\leq m\,.
$$
Locally, we define $\pi^*_\L(\hat{x}^i):=x^i$, for $\hat x^i:=x^i+I_0$. Then the lemma follows from a standard partition of unity argument. 
\end{proof}

\subsection{Semidensities and the $BV$ Operator}\label{SemidensityBVOp}
Let $(\E,\omega)$ be an odd symplectic graded manifold. Clearly, $\omega\wedge\omega=0$, and hence $\omega\wedge(.)$ yields a co-boundary operator on $\Omega(\E)$.
Severa \cite{Severa} recognized the corresponding cohomology classes to be the semidensities, which Khudaverdian \cite{Khudaverdian} had previously introduced as the natural objects to be integrated over Lagrangian submanifolds in $\E$. Furthermore, he gave a cohomological definition of the so-called $BV$-operator, a co-boundary operator acting on semidensities.
In this subsection, we introduce the complex given by semidensities and the $BV$- operator and establish an isomorphism with the de Rham complex of the body of $\E$.
Instead of the spectral sequences used by Severa, we shall use the language of Homological Perturbation Theory ($HPT$), which will also be useful in Section 3.

The basic object in $HPT$ is a {\em contraction}, as introduced by Eilenberg and McLane in the fifties.
It "contracts some big object onto some small object" without loss of cohomology.
In the easiest case, these objects are differential graded modules:
\begin{definition}\label{BVDefcontraction}
A {\em contraction} consists of two differential graded modules $(M,d_M)$ and $(N,d_N)$ over some ring, together with chain maps $\iota: N \rightarrow M$, $p: M \rightarrow N$, i.e., 
$$
d_M\circ\iota=\iota\circ d_N\,,\quad d_N\circ p=p\circ d_M\,,
$$
and a morphism, $h:  M \rightarrow M$, of degree $-1$ such that
\begin{enumerate}
\item $p \circ \iota= \id_N$\,, \label{cond2}
\item $\iota \circ p - \id_M=hd_M+d_Mh$\,,
\item $h^2=h \circ \iota= p \circ h=0$\,. 
\end{enumerate}
Then $p$ is a surjection called the \emph{projection}, $\iota$ is an injection called the \emph{inclusion} and $h$ is called the \emph{homotopy operator}. We write
\begin{equation*}
\xymatrix@1{
(N,d_N) \ar@<3pt>[r]^-\iota & (M,d_M) \quad\quad\quad\ar@<3pt>@{->>}^-p[l]\ar@(dr,ur) & \hspace{-2cm}h\,.}
\end{equation*}
\end{definition}
Condition {\it 2.} implies that the cohomologies of $M$ and $N$ are isomorphic, since it implies that the kernel of $p$ has trivial cohomology:
$$
p\alpha=0\,,\quad d_M\alpha=0\quad\Rightarrow\quad -\alpha=d_M\circ h(\alpha)\,.
$$
Conditions {\it 3.} are also called \emph{side conditions} and can always be satisfied (see, e.g., \cite{LambeStasheff}). 
Lambe and Stasheff \cite{LambeStasheff} adapted this definition to the situation where the objects $M$ and $N$ carry, in addition, an algebra or coalgebra structure. We shall come back to this special case in Section 3.

The basic theorem of $HPT$ is the so-called {\em Perturbation Lemma} and goes back to Brown \cite{RB} and Shih \cite{WS}.
A \emph{perturbation} of the differential $d_M$ is a morphism  $\delta: M \rightarrow M$ of degree $+1$, such that $(d_M + \delta)^2 = 0$.

\pagebreak
\begin{theorem}[Perturbation Lemma] Given a contraction 
\begin{equation*}
\xymatrix@1{
(N,d_N) \ar@<3pt>[r]^-\iota & (M,d_M)\quad\quad \ar@<3pt>@{->>}^-p[l]\ar@(dr,ur) &\hspace{-2cm} h\,}
\end{equation*}
and a perturbation, $\delta$, of $d_M$, with the property that the series below converge, then there is a contraction
\begin{equation*}
\xymatrix@1{
(N,d_N+\tilde\delta) \ar@<3pt>[r]^-{\tilde\iota} & (M,d_M+\delta)\quad\quad\quad\quad \ar@<3pt>@{->>}^-{\tilde p}[l]\ar@(dr,ur) &\hspace{-3cm} \tilde h\,,}
\end{equation*}
where
\begin{eqnarray*}
\tilde\delta & = &\Sum{n\geq 0} p {\delta} (h {\delta})^n \iota\,, \\
\tilde\iota& = & \Sum{n\geq 0} (h {\delta})^n \iota\,, \\
\tilde p & = & \Sum{n\geq 0} p ({\delta} h)^n\,,\\
\tilde h & = & \Sum{n\geq 0} (h{\delta})^n h\,.
\end{eqnarray*}
\end{theorem}
%(Details and references can be found, e.g., in \cite{Real}.)
\begin{proof}
It is straightforward to verify that, formally, conditions {\it 1. - 3.} of Definition \ref{BVDefcontraction} hold for $\tilde p$, $\tilde\iota$, $d_M+\delta$, $d_N+\tilde\delta$ and $\tilde h$. 
\end{proof}

Next, let $(\E,\omega)$ be an $(m,m)$-dimensional odd symplectic manifold.
Since $\omega$ is non-degenerate, there is a two-vector-field $\chi \in  \mathcal X^{2,1}(\E)$ with $\iota_{\chi}\omega=m$.
In local Darboux coordinates (\ref{omindarb}) $\chi$ is given by
\begin{equation*}
\chi =  \frac{\lpartial}{\partial x^i} \wedge \frac{\lpartial}{\partial x_i^\dagger}.
\end{equation*}
For $\eta = f_{i_1,\dots,i_r}^{j_1,\dots,j_s}(x,x^\dag) dx^{i_1} \wedge \ldots \wedge dx^{i_r} \wedge dx_{j_1}^\dagger \wedge \ldots \wedge dx_{j_s}^\dagger \in \Omega(\E)$, we find that
\begin{equation}\label{chiprop}
\big( i_\chi\circ(\omega\wedge)+(\omega\wedge)\circ i_\chi \big) \eta=k(\eta)\eta, 
\end{equation}
where $k(\eta) = m - r + s \in \mathbb N_0$, and we define the homomorphism $h: \Omega(\E) \rightarrow \Omega(\E)$ of $C^{\infty}(\E)$--modules by
\begin{equation}\label{hdef}
h(\eta) := 
\begin{cases}
- \frac{1}{k(\eta)} \iota_\chi(\eta) & {\rm if} \ k(\eta) \neq 0,\\
0 & {\rm otherwise}.
\end{cases}
\end{equation}
Depending on the choice of a projection $\pi:\,\E\rightarrow E_0$, $k$ defines a grading on $\Omega(\E)$,
$$
\Omega(\E)=\bigoplus_{k=0}^{2m}\{\eta\in\Omega(\E)|k(\eta)=k\}\,,
$$
and $k(\eta)=0$ if and only if $\eta$ is in the $C^\infty(\E)$-module $\pi^*(\Omega^m(E_0))$.
Hence there is a contraction
\begin{equation}\label{BVBVcontraction}
\xymatrix@1{
(\pi^\ast(\Omega^m(E_0)),0) \ar@<3pt>@{->}[r]^-{\iota'} & (\Omega(\E),\omega\wedge)\quad\quad\quad\quad \ar@<3pt>@{->>}[l]^-p \ar@(dr,ur) &\hspace{-3cm} h\,,}
\end{equation}
and the differential on the l.h.s. of (\ref{BVBVcontraction}) is zero since $\omega\wedge\eta=0$ for $\eta\in\pi^\ast(\Omega^m(E_0))$. Therefore the l.h.s. of (\ref{BVBVcontraction}) coincides with the cohomology of the r.h.s.:

\begin{equation}\label{BVnormalform}
\pi^*(\Omega^m(E_0)) \stackrel{\cong}\longrightarrow {\rm H}(\Omega(\E),\omega\wedge), \quad
s \longmapsto [s].
\end{equation}
We write $\S(\E):={\rm H}(\Omega(\E),\omega\wedge)$ and call classes of forms in $\S(\E)$ {\em semidensities}.
According to (\ref{BVnormalform}), each semidensity, $[s] \in \S(\E)$, has a representative of the form $f \pi^{\ast} \hat\Omega$ with $f \in C^\infty(\E)$ and $\hat\Omega \in \Omega^m(E_0)$, called the \emph{normal form} of the semidensity $[s]$.

Since $|\omega|=1$ and $d\omega=0$, the de Rham differential $d$ is a perturbation of $\omega\wedge$, i.e., $(\omega\wedge(.) +d)^2=0$, and the Perturbation Lemma yields the contraction
\begin{equation}\label{BVcontraction1pert}
\hspace{2cm}\xymatrix@1{
(\S(\E), \Delta) \cong (\pi^\ast(\Omega^m(E)),\Delta)\ar@<3pt>[r]^-{\tilde\iota}   &(\Omega(\E),\omega\wedge(.) + d)\quad\quad\quad\quad\quad\quad\quad \ar@<3pt>@{->>}[l]^-{\tilde p}\ar@(dr,ur) &\hspace{-5cm} \tilde h\,,}
\end{equation}
with the differential $\Delta$ given by
\begin{equation}\label{Deltasum}
\Delta =p\Sum{n \geq 0} (dh)^n d\iota = pdhd\iota,
\end{equation}
because the summand corresponding to $n = 1$ is the only one that preserves $k(\eta)$ and thus the only one that is not annihilated by $p$. The differential $\Delta$ on $\pi^\ast(\Omega^m(E))$ determines the \emph{BV--operator} on $\S(\E)$, via the isomorphism (\ref{BVnormalform}). To compute the BV--operator on $\S(\E)$, also denoted by $\Delta$, take $[s] \in \S(\E)$ with $s \in \pi^\ast(\Omega^m(E))$. Then $k(s) = 0$ and $k(ds) = 1$, so that (\ref{chiprop}) implies 
\begin{equation}\label{ds}
ds = \iota_\chi(\omega \wedge ds) + \omega \wedge (\iota_\chi ds) = \omega \wedge t, 
\end{equation} 
where $t = \iota_\chi ds$. Noting that $k(t) = 1$, we obtain
\begin{eqnarray}\label{BVBVOp1}
\Delta[s]
& = & [dhd s]\nonumber \\
& = & [dh(\omega \wedge t)] \quad {\rm{by \ (\ref{ds})}} \nonumber\\
& = & - [d(\iota_{\chi}(\omega \wedge t))] \quad \text{by}\quad{\rm{(\ref{hdef})}} \nonumber\\
& = & - [d(t - \omega \wedge(\iota_{\chi} t))] \quad {\rm{by \ (\ref{chiprop})}} \nonumber\\
& = & - [dt]\nonumber\\
& = & - [d(\omega\wedge)^{-1}ds]\quad\text{by (\ref{ds})}\,,
\end{eqnarray}
which coincides (up to a minus sign) with the definition of the $BV$-operator given in \cite{Severa}.

The \emph{$BV$-Laplacian}, $\Delta_{\Omega}$, is a second order graded differential operator on $C^\infty(\E)$ that depends on the choice of pull--back $\Omega = \pi^\ast(\hat\Omega)$ of a volume form $\hat\Omega$ on the body $E_0$, and is determined by
\begin{equation}
\Delta [f \Omega] =: (\Delta_{\Omega}f)[\Omega].
\end{equation}
In local Darboux coordinates, with $\Omega=d x^1\wedge\ldots\wedge d x^m$, a direct calculation yields
\begin{equation*}
\Delta [f\Omega]  = - \frac{\overset{\rightarrow}{\partial}^2f}{\partial  x^i\partial x^\dagger_i} [\Omega]\,,
\end{equation*}
and thus
\begin{equation}\label{BVDelta0}
\Delta_{\Omega} = - \frac{\overset{\rightarrow}{\partial}^2}{\partial  x^i\partial x^\dag_i}.
\end{equation}

What prevents the $BV$-Laplacian from being a derivation on the associative algebra $C^\infty(\E)$ is the odd Poisson bracket:
\bb{BVLemmaDelta1}
\Delta_{\Omega}(f\cdot g)=(\Delta_{\Omega}f)\cdot g+(-1)^{|f|}\{f,g\}+(-1)^{|f|}f\cdot(\Delta_{\Omega}g)\,.
\ee
From (\ref{BVLemmaDelta1}) we deduce, however, that the $BV$-Laplacian is a derivation for the Poisson bracket:
\bb{BVLemmaDelta2}
\Delta_\Omega\{f,g\}=\{\Delta_\Omega f,g\} + (-1)^{|f|+1}\{f,\Delta_\Omega g\}\,.
\ee
In the mathematics literature, a Gerstenhaber algebra whose bracket is derived from a differential according to (\ref{BVLemmaDelta1}) is called a {\em $BV$ algebra}.

Given an odd symplectic manifold $(\E, \omega)$, we may also treat $\omega\wedge$ as a perturbation of the de Rham differential $d$ on $\E$. 
Since $\omega$ is exact, there is a one-form, $\theta$, s.t. $\omega=d\theta$ and since $\omega$ is non-degenerate, we may associate to $\theta$ a vector field $\lder{X}$, which, in local Darboux coordinates, takes the form $\lder{X} = x_i^\dagger \frac{\lpartial}{\partial x_i^\dagger}$.
For $\eta = f_K^{i_1,\dots,i_l,j_1,\dots,j_n}(x) x_{i_1}^\dagger \ldots x_{i_\ell}^\dagger dx^K\wedge dx_{j_1}^\dagger \wedge\dots\wedge dx_{j_n}^\dagger$, with $dx^K = dx^{k_1} \wedge \ldots \wedge dx^{k_q}$, we find that
\begin{equation}
(d \iota_{\lder{X}} + \iota_{\lder{X}} d ) \eta  =  k'(\eta)\eta\,,
\end{equation}
where $k'(\eta) = (n+\ell)$. We define a homomorphism $h': \Omega(\E) \rightarrow \Omega(\E)$ of $C^{\infty}(\E)$--modules by
\begin{equation*}
h'(\eta) = 
\begin{cases}
- \frac{1}{k'(\eta)} \iota_{\lder{X}}(\eta) & {\rm if} \ k'(\eta) \neq 0,\\
0 & {\rm otherwise}\,.
\end{cases}
\end{equation*}
The choice of a projection, $\pi:\E\rightarrow E_0$, together with the natural inclusion $\iota:E_0\rightarrow\E$ induces a contraction
\begin{equation}\label{BVBVcontraction2}
\xymatrix@1{
(\Omega(E_0),d) \ar@<3pt>[r]^-{\pi^*} & (\Omega(\E),d)\quad\quad\quad \ar@<3pt>@{->>}[l]^-{\iota^*} \ar@(dr,ur) &\hspace{-2cm} h'\,,}
\end{equation}
since $k'(\eta)=0$ if and only if $\pi^*\circ\iota^*(\eta)=\eta$.
(Note that $d$ on the l.h.s. of (\ref{BVBVcontraction2}) denotes the de Rham differential on the body.)

We then consider $\omega\wedge$ as a perturbation of $d$ and apply the Perturbation Lemma. Since $\iota^*\circ\omega\wedge=0$ neither the projection $\iota^*$ nor the differential on the l.h.s. of (\ref{BVBVcontraction2}) get modified and we obtain the contraction
\begin{equation}\label{BVBVcontraction3}
\xymatrix@1{
(\Omega(E_0),d) \ar@<3pt>[r]^-{\widetilde{\pi^*}} & (\Omega(\E),d + \omega \wedge)\quad\quad\quad\quad\quad \ar@<3pt>@{->>}[l]^-{\iota^*}\ar@(dr,ur) &\hspace{-4cm} \tilde{h'}\,.}
\end{equation}

Hence, composing the contractions (\ref{BVcontraction1pert}) and (\ref{BVBVcontraction3}), we obtain an isomorphism of differential graded modules,
\begin{equation}\label{Fpi2}
F^\pi: (\S(\E), \Delta) \overset{\cong}{\rightarrow} (\Omega(E_0), d)\,,\quad\text{with}\quad F^\pi:=\iota^*\circ \tilde{\iota'}\,.
\end{equation}
This isomorphism depends on the choice of a projection $\pi:\E\rightarrow E_0$, a fact one is advised to remember.
Note that $(F^\pi)^{-1}=\tilde p\circ\widetilde{\pi^*}$, since
\begin{multline*}
\tilde p\circ\widetilde{\pi^*}\circ\iota^*\circ \tilde{\iota'}
=\tilde p(id+\tilde{h'}(d+\omega\wedge)+(d+\omega\wedge)\tilde{h'})\tilde{\iota'}\\
=id+\Sum{n\geq 0}p(dh)^n\Sum{n\geq 0}(h'\omega\wedge)^nh'(d+\omega\wedge)\Sum{n\geq 0}(hd)^n\iota'\\
+\Sum{n\geq 0}p(dh)^n(d+\omega\wedge)\Sum{n\geq 0}(h'\omega\wedge)^nh'\Sum{n\geq 0}(hd)^n\iota'
=id\,,
\end{multline*}
where the last equality is most easily seen using local Darboux coordinates and keeping track of the number of $dx^\dag$'s that appear.
We also use the isomorphism (\ref{Fpi2}) in order to define a grading on semidensities:
\bb{SDgrading}
|[s]|:=|F^\pi[s]|\,.
\ee
Then $|\Delta|=1$.

For $f = g^{i_1,\dots,i_\ell}(x) x_{i_1}^\dagger \ldots x_{i_\ell}^\dagger$ and $s = f \Omega$, we compute
\begin{equation}\label{Fpi}
F^\pi([s]) = \iota^*\circ \tilde{\iota'} (s)=\iota^*\circ\Sum{n\geq0}(hd)^n\iota'(s) = \iota_{\chi(f)} \Omega,
\end{equation}
where $\chi(f) := g^{i_1,\dots,i_\ell}(x) \frac{\lpartial}{\partial x^{i_1}} \wedge \ldots\wedge \frac{\lpartial}{\partial x^{i_\ell}}$.
From (\ref{Fpi}) we get an isomorphism from the graded functions on $\E$ to the multi-vector fields on the body $E_0$:
\bb{BVchi}
\chi:C^\infty(\E)\overset{\cong}{\longrightarrow} \mathcal X^\bullet(E_0)\,.
\ee
It is then clear that $\chi$ - and hence $F^\pi$ - only depend on the choice of a projection $\pi:\E\rightarrow E_0$, since such a choice allows for an identification of $\E$ with $\Pi T^*E_0$ and, therefore, for an identification of graded functions on $\E$ with multi-vector fields on $E_0$.
The map $\chi$ allows us to identify multiplication with graded functions, $Q$, on the l.h.s. of (\ref{Fpi2}) with interior multiplication, $\iota_{\chi(Q)}$, on the r.h.s. of (\ref{Fpi2})
$$
Q\cdot\mapsto \iota_{\chi(Q)}\,.
$$
Furthermore, we may use it to define the so-called {\em Schouten-Nijenhuis bracket} on multi-vector fields:
\bb{BVSchouten}
[\chi(f),\chi(g)]:=-\chi(\{f,g\})\,.
\ee
If we define the {\em Lie derivative} for semidensities as
\begin{equation}\label{Liedef}
L_Q:=[[Q,\Delta]]=Q\Delta-(-1)^{|Q|}\Delta Q\,,
\end{equation}
we deduce from the well-known relations between the operators $d$, $\iota_\chi$ and $L_\chi$ acting on differential forms corresponding relations between the operators $\Delta$, $Q\cdot$ and $L_Q$ acting on semidensities, which we state in the following proposition.
\begin{proposition}\label{BVPropCR}
For all $Q, Q_1$ and $Q_2$ in $C^\infty(\E)$, the following identities hold:
\begin{enumerate}
\item $[[\Delta,\Delta]]=2\Delta^2=0$, \label{commrel1}
\item $[[Q_1,Q_2]]=0$, \label{commrel2}
\item $[[L_{Q_1},L_{Q_2}]]=-L_{\{Q_1,Q_2\}}$, \label{commrel3}
\item $[[L_{Q_1},Q_2]]=-\{Q_1,Q_2\}$ and \label{commrel4}
\item $[[\Delta,L_Q]]=0$. \label{commrel5}
\end{enumerate}
\end{proposition}
\subsection{BV--Integrals}
Thanks to the isomorphism (\ref{Fpi2}) there is a natural pairing between semidensities and smooth submanifolds $L\subseteq E_0$, once a projection $\pi:\E\rightarrow E_0$ is chosen. In other words, there is a natural pairing between semidensities $[s]$ and pairs $(L,\pi)$:
\bb{BVint0}
\Int{(L,\pi)}[s]:=\Int{L}F^\pi[s]\,.
\ee
Each pair $(L,\pi)$ defines a Lagrangian submanifold, $\L\subset \E$, whose body is $L$ and such that $\pi$ is adapted to $\L$. We show that the r.h.s. of (\ref{BVint0}) only depends on the Lagrangian submanifold $\L$ and not on the particular choice of $\pi$.
Consider an infinitesimal isomorphism of the odd symplectic manifold $(\E,\omega)$ that preserves $\L$ (but not necessarily $\pi$).
Inserting a partition of unity under the integral (\ref{BVint0}) we may work locally and restrict to Hamiltonian variations $\delta_Q$.
Let $\L$ be defined by a sheaf of ideals $I$. Then
$$
\delta_Q\L=0\quad\Leftrightarrow\quad Q\in I\,.
$$
Under such a transformation, the semidensity transforms as $\delta_Q[s]=[L_{\lder{X}_Q}s]$, where $\lder{X}_Q$ is the Hamiltonian vector field associated with $Q$, as defined in (\ref{DGHamVF}), and $L_{\lder{X}_Q}$ denotes the Lie derivative (\ref{BVDefLie}). Since $\omega$ is invariant under Hamiltonian flows, $[L_{\lder{X}_Q}s]$ is well-defined. We need to show that, for $Q\in I$,
\bb{BVint1}
\Int{(L,\pi)}[L_{\lder{X}_Q}s]=0\,,
\ee
which follows from the following lemma.
\begin{lemma}\label{BVLemmaLieder}
For $Q \in C^\infty(\E)$ and $[s]\in\S(\E)$, the Lie-derivative $L_Q$ defined in (\ref{Liedef}) satisfies the equation
\begin{equation*}
L_Q [s] = [L_{\lder X_Q}s].
\end{equation*}
\end{lemma}

\begin{proof}
Consider a semidensity $[s]$ with representative $s\in\Omega(\E)$. Then, as we have seen in (\ref{ds}) and (\ref{BVBVOp1}), $ds=\omega\wedge t$, for some $t\in\Omega(\E)$, and $\Delta[s]=-[dt].$
By (\ref{BVDefLie}) and (\ref{DGHamVF})
\begin{eqnarray*}
[L_{\lder X_Q}s]
& = & [\iota_{\lder X_Q} ds - (-1)^{|Q|}d \iota_{\lder X_Q}s] \\
& = & [\iota_{\lder X_Q} (\omega \wedge t) - (-1)^{|Q|}d \iota_{\lder X_Q}s] \\
& = & [(-1)^{|Q|} (dQ) \wedge t - (-1)^{|Q|}d \iota_{\lder X_Q}s].
\end{eqnarray*}
On the other hand, by (\ref{Liedef}), we find that
\begin{eqnarray*}
L_Q [s]
& = & (Q\Delta -(-1)^{|Q|}\Delta Q)[s] \\
& = & [-Q dt + (-1)^{|Q|} d(\omega \wedge)^{-1}d(Qs) ] \\
& = & [-Q dt + (-1)^{|Q|} d(\omega \wedge)^{-1}(dQ) \wedge s + 
           d(\omega \wedge)^{-1}Q \omega \wedge t] \\
& = & [-Q dt + (-1)^{|Q|} d(\omega \wedge)^{-1}(dQ) \wedge s + (-1)^{|Q|}d(Qt)]\\
& = & [(-1)^{|Q|} d(\omega \wedge)^{-1}(dQ) \wedge s +(-1)^{|Q|}(dQ)\wedge t]\,,
\end{eqnarray*}
and the lemma follows by noticing that
$$
0=\iota_{\lder{X}_Q}(\omega\wedge s)=(-1)^{|Q|}(dQ)\wedge s+(-1)^{|Q|}\omega\wedge\iota_{\lder{X}_Q}s
$$
and hence that
$$
(\omega\wedge)^{-1}(dQ)\wedge s=-\iota_{\lder X_Q}s\,.
$$
\end{proof}
Lemma \ref{BVLemmaLieder} implies
\bb{BVintvar}
\Int{(L,\pi)}[L_{\lder X_Q}s] = \Int{(L,\pi)}Q\Delta[s]-(-1)^{|Q|}\Int{(L,\pi)}\Delta Q[s]=\Int{(L,\pi)}Q\Delta[s]\,,
\ee
where the last eq. follows from (\ref{Fpi2}), (\ref{BVint0}) and the standard Stokes theorem.
Claim (\ref{BVint1}) follows from the observation that, for any semidensity $[s]$,
$$
Q\in I\quad\Rightarrow\quad \Int{(L,\pi)}Q[s]=0\,,
$$
if $I$ defines the submanifold with body $L$ to which $\pi$ is adapted.
Moreover, if $\Delta[s]=0$, we deduce from (\ref{BVintvar}) that the integral (\ref{BVint0}) is preserved under {\em any} Hamiltonian variation of $\L$.
Hence we have proven the following result. 
\begin{theorem}\label{BVThm}
Let $[s]\in\S(\E)$ be a semidensity and $\L\subset\E$ a Lagrangian submanifold with body $L$ (without boundary). Then the following integral is well defined
\begin{equation}\label{BVintegral2}
\int_{\L}[s] := \int_LF^{\pi_\L}[ s]\,,\quad \text{$\pi_\L$ adapted to $\L$}\,.
\end{equation}
Furthermore, if $\Delta[s]=0$ then the integral (\ref{BVintegral2}) is preserved under Hamiltonian variations of $\L$.
\end{theorem}

We conclude this mathematical section with a remark on {\em Berezin integration}.
Let $\M$ be a graded manifold.
W.r.t. local (even/odd) coordinate functions, $\{x^i,\beta^\alpha\}$ ($i=1,\dots,m$, $\alpha=1,\dots,n$), we define a so-called {\em Berezinian}
\bb{Berezinian}
\Omega_B:=f(x,\beta)d\beta^1\wedge\dots\wedge d\beta^n\wedge dx^1\wedge\dots\wedge dx^m \,,
\ee
where $f$ may be expressed in local coordinates as
$$
f(x,\beta)=f_0(x)+f_\alpha(x)\beta^\alpha+\dots +f_{top}(x)\beta^n\dots\beta^1\,.
$$
In the literature, the Berezin-integral of $\Omega_B$ over $\M$ is then defined to be
\bb{BerezinInt}
\Int{\M}\Omega_B:=\Int{M_0}f_{top}(x)dx^1\wedge\dots\wedge dx^m\,,
\ee
where $M_0$ is the body of $\M$.
There is, however, no natural (coordinate-independent) way of characterizing differential forms of the kind (\ref{Berezinian}) on a graded manifold without introducing additional structure.
A natural way of defining "Berezinians" on a graded manifold $\M$ is via semidensities on $\Pi T^*\M$. Then (\ref{Berezinian}) should be replaced by the semidensity
$$
f(x,\beta)[d\beta^\dag_1\wedge\dots\wedge d\beta^\dag_n\wedge dx^1\wedge\dots\wedge dx^m]\,
$$
and (\ref{BerezinInt}) by the $BV$-integral
$$
\Int{\M\subset\Pi T^*\M}f(x,\beta)[d\beta^\dag_1\wedge\dots\wedge d\beta^\dag_n\wedge dx^1\wedge\dots\wedge dx^m]\,,
$$
which coincides with the r.h.s. of (\ref{BerezinInt}) due to (\ref{Fpi}) and (\ref{BVint0}).
%%%%%%%%%%%%%%%%%%%%%%%%%%%%%%%%%%%%%%%%%%%%%%%%%%%%%%%%%%%%%%%%%%%%%%%%%%%%%%%%%%%%%%%

\section{$BV$-quantization}
In this section, we give a general description of how to solve the problem stated in the Introduction using the formalism introduced in Section 2:
For a given Lagrangian system with symmetry, $(M,S_0,P)$, we describe how to construct the $BV$-space, $\E$, (see (\ref{IntroBVspace})) as well as a solution of the $QME$ (\ref{IntroQME}), for a given pull-back, $\Omega$, of a measure on the body of $\E$, such that condition (\ref{Introcond1}) is satisfied. 

The idea is to first extend the original space of fields, $M$, by ghosts - as dictated by the symmetry $P$ - to a graded manifold $\M$ and to pair each field and ghost with an anti-field so as to obtain the odd cotangent bundle\footnote{Note that $M$ is {\em not} the body of $\E$.}
\begin{align}
\Pi T^*&\M=\E\nonumber \\
&\downarrow \pi_0\nonumber \\
&M\,\nonumber
\end{align}
and, second, to find a suitable $\Delta$-closed semidensity, $[s]$, on $\E$, in order to define an expectation, $<f>$, for functions $f\in C^\infty(M)$, by
\bb{BVmeasure}
\langle f\rangle:=Z^{-1}\Int{\L}\pi_0^*f\cdot [s]\,,\quad\text{with}\quad Z:=\Int{\L}[s]\,,
\ee
for a suitable Lagrangian submanifold $\L\subset \E$.
The choice of $\L$ encodes the choice of a gauge.
The semidensity $[s]$ is chosen to be of the normal form
\bb{BVsemidensity}
[s]=\left [ e^{\frac{i}{\hbar}S}\Omega\right]\,,
\ee
where $\Omega$ is the pull-back (w.r.t. some projection) of a measure, $\hat\Omega$, on the body of $\E$, and $S$ is an extension of $S_0$ by ghost- and higher $\hbar$-terms:
\bb{BVS}
S=S_0+\hbar\times \text{ghost terms}+\mathcal O(\hbar ^2)\,.
\ee
The condition $\Delta[s]=0$ is equivalent to
\bb{BVQME}
\frac{1}{2}\{S,S\}-i\hbar\Delta_{\Omega}S=0\,,
\ee
which is the $QME$ (\ref{IntroQME}).
According to Theorem \ref{BVThm}, the first integral in (\ref{BVmeasure}) is {\em gauge-invariant} (i.e., invariant under Hamiltonian variations of $\L$) if 
\bb{BVdeltaDelta}
\Delta (\pi_0^*f[s])=0\quad\Leftrightarrow\quad \delta_{BV} (\pi_0^*f)=0\,,
\ee
where
\bb{BVdeltaBV}
\delta_{BV} :=\{S,.\}-i\hbar \Delta_{\Omega}\,.
\ee
The $QME$ implies that $\delta_{BV}^2=0$, and hence (\ref{BVmeasure}) defines a measure on the cohomology classes of $\delta_{BV}$.
Note that the $\delta_{BV}$-cohomology is isomorphic to the $\Delta$-cohomology and hence to the deRham cohomology of the body of $\E$. Since the body of $\E$ is a vector bundle over $M$, the $\delta_{BV}$-cohomology is nothing but the {\em deRham cohomology of $M$}.

The information encoded in the action $S_0$, the symmetry $P$ and the measure $\hat\Omega$ on the body enables us to split $\delta_{BV}$ according to (\ref{BVdeltaBV}) and (\ref{BVS}) and to reinterpret the $\delta_{BV}$-cohomology by using Homological Perturbation Theory ($HPT$):
Under a regularity condition (see Sect. 3.1), the cohomology of the differential $\delta_0=\{S_0,.\}$ on $C^\infty(\E)$, given by the first term in (\ref{BVS}), coincides with the {\em restriction of graded functions to the shell}.
A solution of the classical master equation ($CME$) $\{S,S\}=0$, given by adding the ghost terms in (\ref{BVS}) to $S_0$, can be seen as a perturbation of $\delta_0$ that induces the {\em $BRST$-differential on the shell}.
We will see that the $BRST$-cohomology at ghost degree zero contains the {\em classical observables}, i.e., the gauge-invariant functions on the shell, and that the $BRST$-cohomology at ghost degree one contains the {\em anomalies}, i.e., obstructions to solving the $QME$.

Perturbation of the $BRST$-cohomology by $\Delta_\Omega$ in (\ref{BVdeltaBV}) then enables us to construct an {\em invariant effective measure on the shell} that contains the off-shell contributions to (\ref{BVmeasure}) as a formal power series in $\hbar$. In general, this power series in $\hbar$ does not converge.
If it terminates, we say that the integral {\em localizes} on the shell. As an example, we will rederive Duistermaat-Heckman localization in the $BV$-formalism.
If the measure $\hat\Omega$ is not invariant under the symmetry, we need to add higher order $\hbar$-terms to $S$ - as indicated in (\ref{BVS}) - in order to obtain a solution of the $QME$.
We will find obstructions to solving the $QME$ (anomalies) as well as obstructions to solving the $CME$ for the case of an open symmetry in the $BRST$-cohomology at ghost degree one.

\subsection{The shell}
We start with a discussion of the cohomology generated by the first term, $S_0$, in (\ref{BVS}).
The ghosts being spectators, we may restrict our attention to the space $\E_0:=\Pi T^*M$, where the fibres carry ghost degree one, and consider the differential
\bb{BVdelta0}
\delta_0:=\{S_0,.\}\,
\ee
thereon.
%If we assign to the fibres of $\Pi T^*M$ {\em anti-ghost degree} one, the differential (\ref{BVdelta0}) decreases the anti-ghost degree by one.
The set of critical points of $S_0$ is denoted by $\Sigma$ and is called the {\em shell}.
We impose the following {\em regularity condition} on $S_0$:
On the shell, the Hessian of $S_0$ is supposed to be of constant rank $m-k$, where $m$ is the dimension of $M$ and $k$ is the dimension of the shell:
\bb{BVRegularityCond}
rk(S_{0,ij})_{\arrowvert\Sigma}=m-k\,.
\ee
This is to say that any function $f\in C^\infty(M)$ that vanishes on-shell can be written as a linear combination of equations of motion $S_{0,i}$ (with coefficients in $C^\infty(M)$), which, in turn, amounts to saying that $f$ is $\delta_0$-exact:
$$
f=\Sum{i}S_{0,i}g^i=\{S_0,x^\dag_ig^i\}\,.
$$
Hence, at ghost degree $0$ the cohomology of $\delta_0$ is the restriction of graded functions to the shell.
At negative ghost degree $\delta_0(f)=0$ implies that the multi-vector field associated with $f$ is tangential to the shell. Hence,
\bb{BVdelta0coho}
H(C^\infty(\E_0),\delta_0)=C^\infty(\Pi T^*\Sigma)\,.
\ee

\subsection{Closed irreducible symmetries}
Next, we introduce ghosts, as dictated by the symmetry $P$, and perturb the $\delta_0$-cohomology by adding ghost terms to $S_0$.
In this subsection we consider {\em closed symmetries}, i.e., symmetries where $E=0$ in (\ref{IntroP2}). Due to the assumption that $P$ is finitely generated, there is a vector bundle $E_1$ and a bundle map $\rho_1$
\bb{BVghostbundle}
\xymatrix{
E_1 \ar[rr]^{\rho_1} \ar@{->>}[dr] && TM \ar@{->>}[dl] \\
& M &}
\ee
such that
$$
\rho_1(\Gamma(E_1))=P\,.
$$
We call the bundle $E_1$ a {\em ghost bundle}.
Assigning ghost degree $-1$ to the fibres of $E_1$ yields the graded vector bundle denoted by $E_1[-1]$.
We denote the associated graded manifold by $\M$. According to (\ref{gradedmfd}), local graded functions on $\M$ are generated by local sections, $\beta^\alpha$, of $E_1^*[1]$ - the {\em ghosts} that carry degree one. 
The symmetry $P$ is said to be {\em irreducible} if there is a ghost bundle (\ref{BVghostbundle}) such that $\rho_1:\Gamma(E_1)\rightarrow P$ is injective.
Then the tensor $T$ in (\ref{IntroP2}) equips $E_1$ with the structure of a {\em Lie algebroid}.
Pairing each field (coordinate function on $M$) and ghost with an anti-field extends $\M$ to the graded manifold $\E:=\Pi T^*\M$, which is associated with the graded vector bundle $T^*[1]M\oplus E_1[-1]\oplus E_1^*[2]$.
Graded functions on $\E$ are then generated, locally, by the coordinates $\{x^\dag_i,\beta^\alpha,\beta^\dag_\alpha\}$ of degrees $-1$, $1$ and $-2$, respectively, with coefficients in $C^\infty(M)$.
Then the natural odd symplectic structure on $\E$, $\omega=dx^i\wedge dx^\dag_i+d\beta^\alpha\wedge d\beta^\dag_\alpha$, has ghost degree $-1$.
The body of $\E$ is $E_1^*[2]$. From the theory of Lie algebroids we know that the dual of a Lie algebroid is equipped with a natural Poisson bracket, i.e., with a two-vector field $P_1\in\Gamma\Lambda^2TE_1^*$ with the property that $[P_1,P_1]=0$, where $[.,.]$ denotes the Schouten-bracket.
Due to the isomorphism (\ref{BVchi}) and eq. (\ref{BVSchouten}), $P_1$ corresponds to a graded function, $S_1$, on $\E$ with $\{S_1,S_1\}=0$.
The section of $S_\bbR(T[-1]M\oplus E_1^*[1]\oplus E_1[-2])$ associated with $S_1$ is given by the sum of the two terms $\rho_1\in E_1^*\otimes TM$ and $-\rho_1^{-1}[\rho_1\wedge\rho_1]/2\in E_1^*\wedge E_1^*\otimes E_1$. What this means becomes clear in local coordinates, in which the graded function $S_1$ reads
\bb{BVS1}
S_1=x^\dag_i\rho_{1\alpha}^i(x)\beta^\alpha-\beta^\dag_\alpha T^\alpha_{\beta\gamma}(x)\beta^\beta\beta^\gamma\,,\quad
T^\alpha_{\beta\gamma}(x):=(\rho^{-1}_1)^\alpha_{i}(x)\rho^j_{1\beta}(x)\rho^i_{1\gamma,j}(x)\,.
\ee
The restriction of the differential 
\bb{BVdelta1}
\delta_1:=\{S_1,.\}
\ee
to graded functions on $\M$ is usually called $BRST$-differential.
If the symmetry arises from a Lie group acting on $M$, the $BRST$-cohomology coincides with the corresponding Lie algebra cohomology with coefficients in $C^\infty(M)$.
In the open case (see Sect. 3.4), $\delta_1$ is not a differential but still a perturbation of $\delta_0$, i.e., the $BRST$-differential is only defined on the shell.
The $BRST$-differential on the shell contains interesting physical information: the {\em classical observables} (gauge-invariant functions on the shell) at ghost degree zero and the {\em anomalies} at ghost degree $1$ (see Subsect. \ref{Anomalies}).

\subsection{Closed reducible symmetries}
In the {\em reducible} case, replacing $\rho_1^{-1}$ in (\ref{BVS1}) by an arbitrary injection
\bb{BVrediota}
\iota:\, P\hookrightarrow \Gamma(E_1)\,,\quad\text{with}\quad \rho_1\circ\iota =id\,,
\ee
we again obtain a coboundary operator $\delta_1$, which produces, however, extra cycles, corresponding to $\delta_1$-closed graded functions that are associated with one-forms $\eta\in\Gamma E_1^*$ that vanish on the image of $\iota$. (If the one-form associated with the graded function $\delta_1(f)$, for any $f\in C^\infty(M)$, were to vanish on the image of $\iota$, $\delta_1(f)$ would have to vanish as well.)
In order to kill these cycles, we need to introduce second-order ghosts. Iterating this process, we end up with a sequence of ghost bundles
\bb{BVhoghostbundle}
\xymatrix{
\dots\,E_2 \ar[r]^{\rho_2} \ar@{->>}[dr] &E_1\ar[r]^{\rho_1} \ar@{->>}[d]& TM \ar@{->>}[dl] \\
& M &}
\ee
such that
$$
\rho_i(\Gamma(E_i))=\operatorname{Ker}(\rho_{i-1})\,.
$$
The graded bundle $E_1[-1]\oplus E_2[-2]\oplus\dots$ defines a graded manifold $\M$. We then find a natural graded function $S_1$, with $\{S_1,S_1\}=0$, on $\E=\Pi T^*\M$. Detailed descriptions of how to construct $S_1$ in the presence of higher order ghosts can be found, e.g., in \cite{Gomis}.

\subsection{Open symmetries}\label{BVopen}
In this subsection, we consider an {\em open} gauge algebra, i.e., a symmetry $P$ with $E\neq 0$ in (\ref{IntroP2}).
As usual, we write $\delta_0:=\{S_0,.\}$ and choose a contraction
\bb{BVcontraction}
\xymatrix@1{
(H(C^\infty(\E)),0)\ar^{\iota\quad}[r] & (C^\infty(\E),\delta_0)\quad\quad\quad\quad\ar@<2pt>@{->>}^{p\quad}[l]\ar@(dr,ur) &\hspace{-3cm}h\,,
}
\ee
where $h$ satisfies
\bb{BVOpenHomotopy}
h\delta_0+\delta_0h=\iota p-id\,
\ee
as well as the side conditions
\bb{BVopensidecond}
h^2=h\circ\iota=p\circ h=0\,.
\ee

The data (\ref{IntroP2}) yield an even function $S_1\in C^\infty(\E)$ with the properties
\bb{BVopenS1}
\delta_0 S_1=0\,,\quad \{S_1, S_1\}=\delta_0\text{-exact}\,.
\ee
We attempt to find a solution of the $CME$ $\{S,S\}=0$ of the form
$$
S=\Sum{m\geq 0}S_m\,.
$$
The $CME$ can be written in the form of the Maurer-Cartan equation
\bb{BVopenMC}
\delta_0\left(\Sum{m\geq 1}S_m\right)+\frac{1}{2}\left\{\Sum{m\geq 1}S_m,\Sum{m\geq 1}S_m\right\}=0\,.
\ee
Following ideas of Kontsevich \cite{Kontsevich}, we want to treat $\{.,.\}$ as a perturbation of $\delta_0$.
Therefore, we need to extend the contraction (\ref{BVcontraction}) to a contraction of graded anti-commutative algebras as defined in (\ref{DGasymmalg}):
\bb{BVcontraction1}
\xymatrix@1{
(\bigwedge_\bbR(H(C^\infty(\E)))[1],0)\ar^{\iota\quad\quad}[r] & (\bigwedge_\bbR(C^\infty(\E))[1],\delta_0)\quad\quad\quad\quad\quad\quad\quad\ar@<2pt>@{->>}^{p\quad\quad}[l]\ar@(dr,ur) &\hspace{-5cm}h\,.
}
\ee
We use the symbol $\wedge$ for the associative products in $\bigwedge_\bbR(C^\infty(\E))[1]$ and $\bigwedge_\bbR(H(C^\infty(\E)))[1]$.
The $[1]$ indicates that the anti-commutative product is taken w.r.t. the ghost gradings in $C^\infty(\E)$ and $H(C^\infty(\E))$ shifted by one, i.e.,
\bb{BVWedge}
f\wedge g=-(-1)^{(|f|+1)(|g|+1)}g\wedge f\,.
\ee
In order for (\ref{BVcontraction1}) to be a contraction, we extend the definitions of $\iota$, $p$, $\delta_0$ and $h$ as follows:
\begin{align*}
\iota(f_1\wedge\dots\wedge f_k)&=\iota(f_1)\wedge\dots\wedge \iota(f_k)\,,\\
p(f_1\wedge\dots\wedge f_k)&=p(f_1)\wedge\dots\wedge p(f_k)\,,\\
\delta_0(f_1\wedge\dots\wedge f_k)&=\sum_{i=1}^k(-1)^{|f_1|+\dots+|f_{i-1}|+i+k}f_1\wedge\dots\wedge (\delta_0f_i)\wedge\dots\wedge f_k\,,\\
h(f_1\wedge\dots\wedge f_k)&=\sum_{i=1}^k(-1)^{|f_1|+\dots+|f_{i-1}|+i+k}f_1\wedge\dots\wedge (hf_i)\wedge \iota p(f_{i+1})\wedge\dots\wedge\iota p(f_k)\,.
\end{align*}
Eqs (\ref{BVOpenHomotopy}) and (\ref{BVopensidecond}) then hold on $\bigwedge_\bbR (C^\infty(\E))[1]$, e.g.,
\begin{align*}
(h\delta_0+\delta_0h)(f_1\wedge\dots\wedge f_k)
&=\sum\limits_{i=1}^kf_1\wedge\dots\wedge(h\delta_0+\delta_0h)f_i\wedge\iota pf_{i+1}\wedge\dots\wedge\iota pf_k\\
&=(\iota p-\id)(f_1\wedge\dots\wedge f_k)\,,
\end{align*}
where the first equation is derived from the fact that $h$ and $\delta_0$ both change the ghost degree by one, and, in the second equation, (\ref{BVOpenHomotopy}) on $C^\infty(\E)$ has been used.
The contraction (\ref{BVcontraction1}) is an example of a contraction of {\em co-algebras}, as studied by Gugenheim, Lambe and Stasheff in \cite{GLS}.
Now, the bracket 
$$
\{.,.\}:\,C^\infty(\E)\wedge C^\infty(\E)\longrightarrow C^\infty(\E)
$$
can be considered as a perturbation of $\delta_0$ on $\bigwedge_\bbR (C^\infty(\E))[1]$.
Therefore, we define
$$
\{.,.\}(f_1\wedge\dots\wedge f_k)=\Sum{\sigma}(-1)^{\sigma+k} \{f_{i_1},f_{i_2}\}\wedge f_{i_3}\wedge\dots\wedge f_{i_k}\,,
$$
where the sum ranges over all $(2,k-2)$ shuffles, i.e., over all permutations of $k$ letters that leave the order of the first two and the order of the last $k-2$ letters unchanged. The symbol $(-1)^\sigma$ denotes the sign that is associated with the permutation $\sigma$ according to (\ref{BVWedge}).
For a triple, e.g., we get the formula
\begin{multline*}
\{.,.\}(f_1\wedge f_2\wedge f_3)=\\
-\{f_1,f_2\}\wedge f_3-(-1)^{(|f_2|+1)(|f_3|+1)+1}\{f_1,f_3\}\wedge f_2 - (-1)^{(|f_1|+1)(|f_2|+|f_3|)}\{f_2,f_3\}\wedge f_1\,.
\end{multline*}
The fact that $\{.,.\}$ is a perturbation of $\delta_0$ follows from the graded Leibniz- and Jacobi- identities of the bracket, which may be written in terms of operators on $\bigwedge_\bbR (C^\infty(\E))[1]$ as
$$
\delta_0\{.,.\}+\{.,.\}\delta_0=\{.,.\}^2=0\,.
$$
Note that the sign $(-1)^k$ in the definition of $\delta_0$ is needed for $\delta_0$ and $\{.,.\}$ to {\em anti-}commute.

Applying the Perturbation Lemma yields a co-boundary operator, $D$, on $\bigwedge_\bbR (H(C^\infty(\E)))[1]$ of the form
\bb{BVPLD}
D=\Sum{n\geq 2}D_n\,,\quad D_n=p\{.,.\}(h\{.,.\})^{n-2}\iota:\,H(C^\infty(\E))^{\wedge n}\longrightarrow H(C^\infty(\E))\,
\ee
and a new contraction
\bb{BVcontraction5}
\hspace{2.5cm}
\xymatrix@1{
(\bigwedge_\bbR(H(C^\infty(\E)))[1],D)\ar^{\tilde\iota\quad\quad\quad\quad\quad\quad}[r] & (\bigwedge_\bbR(C^\infty(\E))[1],\delta_0+\{.,.\})\quad\quad\quad\quad\quad\quad\quad\quad\quad\quad\ar@<2pt>@{->>}^{\tilde p\quad\quad\quad\quad\quad\quad}[l]\ar@(dr,ur) &\hspace{-7cm}\tilde h\,.
}
\ee
%Notice that $D_2^2=0$ and thus $D_2$ equips $H(C^\infty(\E))$ with the structure of a Gerstenhaber algebra.
The new inclusion $\tilde\iota$ reads
\bb{BVPLiota}
\tilde\iota=\Sum{n\geq 1}\tilde\iota_n\,,\quad \tilde\iota_n=(h\{.,.\})^{n-1}\iota\,
\ee
and the fact that it is a chain map is expressed by
\bb{BVOpenIotaChain}
(\delta_0+\{.,.\})\circ\tilde\iota=\tilde\iota\circ D\,.
\ee
We will need the formulae
\begin{multline}\label{BVpreKontsevich}
\tilde\iota_k(a^{\wedge m})=
\frac{1}{m-k+1}\sum\limits_{i=1}^k\begin{pmatrix}m\\i\end{pmatrix}\tilde\iota_{k+1-i}(a^{\wedge(m-i)})\wedge\tilde\iota_i(a^{\wedge i})\\
=\Sum{i_1+\dots+ i_{k'}=m}\frac{1}{(k')!}\begin{pmatrix}m\\i_1\end{pmatrix}\begin{pmatrix}m-i_1\\i_2\end{pmatrix}\dots\begin{pmatrix}i_{k'}\\i_{k'}\end{pmatrix}\tilde\iota_{i_1}(a^{\wedge i_1})\wedge\dots\wedge\tilde\iota_{i_{k'}}(a^{\wedge i_{k'}})\,,\quad(k'=m-k+1)\,,
\end{multline}
for $k\leq m$, which hold for any $a\in C^\infty(\E)$ of {\em even} ghost degree and are proven by induction.
The new projection $\tilde p$ reads
\bb{BVPLp}
\tilde p=\Sum{n\geq 1}\tilde p_n\,\quad \tilde p_n=p(\{.,.\}h)^{n-1}\,
\ee
and the fact that it is a chain map is expressed by
\bb{BVOpenpChain}
\tilde p\circ(\delta_0+\{.,.\})=D\circ\tilde p\,.
\ee

Eqs (\ref{BVopenS1}) can be expressed as
\bb{BVopenr}
(\delta_0+\{.,.\}) S_1\wedge S_1=\delta_0\text{-exact}\,.
\ee
By applying $\tilde p$ on both sides of (\ref{BVopenr}) and using (\ref{BVOpenpChain}), we get that
$$
\tilde p(\delta_0+\{.,.\})( S_1\wedge S_1)=D\tilde p( S_1\wedge S_1)=D_2\tilde p_1(S_1\wedge S_1)=D_2(p S_1)^{\wedge 2}=0\,.
$$
We will show below that the condition
\bb{BVopencond}
D_k(pS_1)^{\wedge k}=0
\ee
is satisfied, for all $k\geq 2$, if a certain cohomology class vanishes.
Then we have that
\bb{BVDexp}
D\left(\Sum{n\geq 1}\frac{1}{n!}(pS_1)^{\wedge n}\right)=0\,.
\ee
Furthermore, we deduce from eqs (\ref{BVpreKontsevich}) that
\bb{BVKontsevich}
\tilde\iota\left(\Sum{n\geq 1}\frac{1}{n!}(pS_1)^{\wedge n}\right)=\Sum{n\geq 1}\frac{1}{n!}\left(\Sum{m\geq 1}\frac{1}{m!}\tilde\iota_m((pS_1)^{\wedge m})\right)^{\wedge n}\,.
\ee
If we choose the contraction (\ref{BVcontraction}) in such a way that 
\bb{BVopencond0}
h(S_1)=0\,,
\ee
which, due to (\ref{BVOpenHomotopy}), implies that
\bb{BVopencond0'}
\iota p(S_1)=S_1\,,
\ee
we may set
\bb{BVSm}
S_m=\frac{1}{m!}\tilde\iota_m(pS_1)^{\wedge m}\quad (m\geq 1)\,.
\ee
From eqs (\ref{BVOpenIotaChain}), (\ref{BVDexp}) and (\ref{BVKontsevich}) we obtain
\bb{BVpreCME}
(\delta_0+\{.,.\})\left(\Sum{n\geq 1}\frac{1}{n!}\left(\Sum{m\geq 1}S_m\right)^{\wedge n}\right)=0\,,
\ee
which implies that $S:=\Sum{m\geq 0}S_m$ is a solution of the $CME$, written in the form (\ref{BVopenMC}):
$$
\delta_0\left(\Sum{m\geq 1}S_m\right)+\frac{1}{2}\left\{\Sum{m\geq 1}S_m,\Sum{m\geq 1}S_m\right\}=0\,.
$$

In the remainder of this subsection we show that contraction (\ref{BVcontraction}) can be chosen such that condition (\ref{BVopencond}) is satisfied, for all $k\geq 2$, if the first cohomology class of the $BRST$ operator, defined on the shell, vanishes.
First, use eqs (\ref{BVpreKontsevich}) in order to express the l.h.s. of (\ref{BVopencond}) as
\bb{BVopencond2}
D_k(pS_1)^{\wedge k}=p\{.,.\}\tilde\iota_{k-1}(pS_1)^{\wedge k}=\frac{k!}{2}p\Sum{i_1+i_2=k}\{S_{i_1},S_{i_2}\}\,.
\ee
If we define
\bb{BVTk}
T_k:=\frac{k!}{2}\Sum{i_1+i_2=k}\{S_{i_1},S_{i_2}\}=\{.,.\}\tilde\iota_{k-1}(pS_1)^{\wedge k} \,,
\ee
conditions (\ref{BVopencond}) thus read
\bb{BVopencond3}
pT_k=0\,,\quad k\geq 2\,.
\ee
Furthermore, we deduce from (\ref{BVSm}) and (\ref{BVTk}) that
\bb{BVopenST}
S_k=\frac{1}{k!}h(T_k)\,,\quad k\geq 2\,.
\ee
%and, since $h^2=0$, we have that
%\bb{BVopenhS}
%h(S_k)=0\,,\quad k\geq 1\,.
%\ee
Now assume that conditions (\ref{BVopencond3}) are satisfied up to order $n$:
\bb{BVOpenCondInd}
p(T_k)=0\,,\quad2\leq k\leq n\,.
\ee
We show that this implies that
\bb{BVdelta0T}
\delta_0T_k=0\,,\quad 2\leq k\leq n+1\,.
\ee
Using eqs (\ref{BVOpenHomotopy}), (\ref{BVPLiota}), (\ref{BVTk}) and the Jacobi identity $\{.,.\}^2=0$, we see that $\delta_0T_k=0$ if the following conditions are met
$$
p\{.,.\}\tilde\iota_{i}(pS_1)^{\wedge k}=0\,,\quad i=1,\dots, k-2\,.
$$
But these conditions follow from the conditions (\ref{BVOpenCondInd}), as an application of eq. (\ref{BVpreKontsevich}) shows:
\begin{align*}
p\{.,.\}\tilde\iota_{i}(pS_1)^{\wedge k}
&=p\{.,.\}\frac{k!}{(k')!}\Sum{i_1+\dots+i_{k'}=k}S_{i_1}\wedge\dots\wedge S_{i_{k'}}\quad (k'=k-i+1)\\
&=\frac{(-1)^{k'+1}k!}{(k'-2)!}\frac{1}{2}\Sum{i_1+\dots+i_{k'}=k}p\{S_{i_1},S_{i_2}\}\wedge pS_{i_3}\wedge\dots\wedge pS_{i_{k'}}\\
&=\frac{(-1)^{k'+1}k!}{(k'-2)!}\Sum{I+i_3+\dots+i_{k'}=k}\frac{1}{I!}pT_I\wedge pS_{i_3}\wedge\dots\wedge pS_{i_{k'}}=0\,.
\end{align*}
Notice that we are free to define $h$ and, because of (\ref{BVopenST}), $S_k$, for $k\geq 2$, as long as eqs (\ref{BVOpenHomotopy}) and (\ref{BVopensidecond}) are satisfied.
Using eqs (\ref{BVopenST}), (\ref{BVOpenCondInd}) and (\ref{BVdelta0T}) we see that eqs (\ref{BVOpenHomotopy}) and (\ref{BVopensidecond}) lead to the following conditions for $S_k$:
\bb{BVOpenDelta0T}
\delta_0S_k=-\frac{1}{k!}T_k\,,\quad 2\leq k\leq n\,,
\ee
and 
\bb{BVopenpSk}
pS_k=0\,,\quad 2\leq k\leq n\,.
\ee
Forgetting condition (\ref{BVopenpSk}) for a moment, we see that $S_n$ is only determined by $S_k$, for $1\leq k< n$, up to a $\delta_0$-closed term.
We need to show that condition (\ref{BVOpenCondInd}) is satisfied for $k=n+1$. Due to (\ref{BVdelta0T}) this is the same as to show that $T_{n+1}$ is $\delta_0$-exact.
%Applying (\ref{BVOpenHomotopy}) to $S_k$ and using (\ref{BVopenST}), (\ref{BVopenhS}) and (\ref{BVOpenDelta0T}), we find that 
%\bb{BVopenpS}
%pS_k=0\,,\quad k\leq n\,.
%\ee
We show that the freedom of adding a $\delta_0$-closed term to $S_{n}$ is enough to render $T_{n+1}$ $\delta_0$-exact, if the on-shell $BRST$-cohomology at ghost degree one vanishes.
Thus, denote the $BRST$ co-boundary operator on $H(C^\infty(\E))$ by 
\bb{BVopenBRST}
\delta_{BRST}:=p\{S_1,.\}\iota\,,
\ee
and check that it squares to zero.
We show that $pT_{n+1}$ is $\delta_{BRST}$-closed if conditions (\ref{BVOpenCondInd}) hold:
\begin{align*}
\delta_{BRST}pT_{n+1}
&=(n+1)!p\Sum{i_1+i_2=n+1}\{\{S_1,S_{i_1}\},S_{i_2}\}\\
&=(n+1)!p\Bigg(\Sum{\overset{i_1+i_2=n+1}{i_1,i_2\geq 2}}\{\{S_1,S_{i_1}\},S_{i_2}\}+\{\{S_1,S_1\},S_n\}+\{\{S_1,S_n\},S_1\}\Bigg)\\
&=-(n+1)!p\Bigg(\Sum{\overset{i_1+i_2=n+1}{i_1,i_2\geq 2}}\Bigg(\{\delta_0S_{i_1+1},S_{i_2}\}+\frac{1}{2}\Sum{\overset{j_1+j_2=i_1+1}{j_1,j_2\geq 2}}\{\{S_{j_1},S_{j_2}\},S_{i_2}\}\Bigg)\\
&\hspace{9.5cm}-\{S_n,\delta_0S_2\}\Bigg)\\
&=-(n+1)!p\Bigg(\delta_0\Sum{\overset{i_1+i_2=n+2}{i_1,i_2\geq 2}}\{S_{i_1},S_{i_2}\}+\frac{1}{2}\Sum{\overset{j_1+j_2+i_2=n+2}{j_1,j_2,i_2\geq 2}}\{\{S_{j_1},S_{j_2}\},S_{i_2}\}\Bigg)=0
\end{align*}
where we used (\ref{BVOpenHomotopy}), (\ref{BVTk}), (\ref{BVdelta0T}) and $p\{S_1,.\}\delta_0=0$ in the first equation, (\ref{BVTk}) and (\ref{BVOpenDelta0T}) in the third and the Jacobi identity as well as $p\circ\delta_0=0$ in the last equation.
If the $\delta_{BRST}$-cohomology vanishes at ghost degree one we can thus write
$$
T_{n+1}=\{S_1,G\}+\delta_0-\text{exact}\,,
$$
for some $\delta_0$-closed term $G$.
And, since
$$
T_{n+1}=k!\{S_1,S_{n}\}+\text{terms independent of $S_{n}$}\,,
$$
we may use the freedom of adding a $\delta_0$-closed term to $S_{n}$ to render $T_{n+1}$ $\delta_0$-exact.
Finally, condition (\ref{BVopenpSk}) can always be satisfied without spoiling the $\delta_0$-exactness of $T_{n+1}$ by replacing $S_n$ by $S_n-\iota\circ pS_n$.
Hence conditions (\ref{BVopencond}) are satisfied for all $k$, which was to be shown.

For a more general view on this subject we refer to \cite{SH}.

%%%%%%%%%%%%%%%%%%%%%%%%%%%%%%%%%%%%%%%%%%%%%%%%%%
\subsection{The $BV$-Laplacian and localization}
The last three subsections were devoted to finding a solution of the $CME$ $\{S,S\}=0$.
In this subsection, we assume that there are no symmetries, i.e. $S=S_0$, and perturb the differential $\delta_0:=\{S_0,.\}$ by the $BV$-Laplacian so as to obtain the $BV$-differential
$$
\delta_{BV}=\delta_0-i\hbar\Delta_\Omega\,;
$$
see (\ref{BVdeltaBV}).
We restrict our attention to $\E_0:=\Pi T^*M$ and the differential $\delta_0:=\{S_0,.\}$ and choose a contraction
\bb{BVcontraction3}
\xymatrix@1{
(H(\delta_0),0)\ar^{\iota\quad\quad}[r] & (C^\infty(\E_0),\delta_0)\quad\quad\quad\quad\ar@<2pt>@{->>}^{p\quad\quad}[l]\ar@(dr,ur) &\hspace{-3cm}h\,.
}
\ee
Under the regularity condition (\ref{BVRegularityCond}) we have seen in (\ref{BVdelta0coho}) that $H(\delta_0)=C^\infty(\Pi T^*\Sigma)$.
Then the projection $p$ is simply the restriction to $\Sigma$ and we can choose $\iota$ to extend a super-function on $\Pi T^*\Sigma$ to a smooth super-function supported on a $\epsilon$-neighbourhood of $\Sigma$ in $M$. We shall keep track of the $\epsilon$-dependence of $\iota$ and write $\iota^\epsilon$ instead. 
We perturb $\delta_0$ by $-i\hbar\Delta_\Omega$ - for some pulled-back measure $\Omega$ on $M$ - 
and apply the Perturbation Lemma in order to get a pertubed contraction with inclusion, $\tilde\iota$, projection, $\tilde p$, and homotopy operator, $\tilde h$, given by formal power series in $\hbar$.
Applying equation {\it 2.} in Definition \ref{BVDefcontraction} (see Sect. \ref{SemidensityBVOp}) of the perturbed contraction to the constant function $1$ on $M$ gives
$$
\tilde\iota^\epsilon\tilde p(1)-1=\frac{i}{\hbar}\delta_{BV}\tilde h(1)\,.
$$
Since
$$
\Delta\left[\tilde h(1)e^{iS_0/\hbar}\Omega\right]=\frac{i}{\hbar}\delta_{BV}(\tilde h(1))\left[e^{iS_0/\hbar}\Omega\right]\,
$$
and integrals of $\Delta$-exact semidensities vanish, we get that
\bb{BVLocInt}
\Int{M}e^{iS_0/\hbar}[\Omega]=\Int{M}\tilde\iota^\epsilon\tilde p (1)e^{iS_0/\hbar}[\Omega].
\ee
Note, that $\tilde\iota^\epsilon\tilde p (1)$ is localized in an $\epsilon$-neighbourhood of $\Sigma$, since $\Delta_\Omega(\iota^\epsilon\tilde p (1))=0$ and hence $\tilde\iota^\epsilon\tilde{p} (1)$ can be replaced by $\iota^\epsilon\tilde{p} (1)$.
In the limit $\epsilon\rightarrow 0$ the integrand on the r.h.s. of eq. (\ref{BVLocInt}) tends to an effective measure on $\Sigma$ in the form of a formal power series in $\hbar$. In general, the latter does not converge. As an example where it terminates, we rederive the Duistermaat-Heckman localization formula.
Before we do so, we describe an equivalent way of constructing an effective measure on $\Sigma$.
We remove the critical points from $M$: $\tilde M:=M\backslash\Sigma$, $\tilde\E_0:=\Pi T^*\tilde M$ and choose a contraction
\bb{BVcontraction4}
\xymatrix{
0\ar^{\iota\quad\quad\quad\quad}[r] & *{(C^\infty(\tilde\E_0),\delta_0)\quad\quad\quad\quad}\ar@<2pt>@{->>}^{p\quad\quad\quad\quad}[l]\ar@(dr,ur) &\hspace{-3cm}h\,.
}
\ee
The operator $h$ is multiplication with a super-function, denoted $h$ as well, of anti-ghost-degree $-1$ that satisfies
$$
\delta_0(h)=1\,.
$$
Such a function exists on $\tilde \E_0$.

Now, we perturb the differential $\delta_0$ by $-i\hbar\Delta_\Omega$.
The Perturbation Lemma (or a direct computation) then tells us that
$$
id=\tilde h\delta_{BV}+\delta_{BV}\tilde h\,,
$$
where $\tilde h$ is given by the formal power series 
\bb{DHPL}
\tilde h:=\Sum{n\geq 0}(-i\hbar)^{n}(h\Delta_\Omega)^nh\,.
\ee
Assuming that $\Sigma$ has zero measure w.r.t. $\Omega$, we obtain a formula of the kind (\ref{BVLocInt})
\bb{DHIntexact}
\Int{M} e^{iS_0/\hbar}[\Omega]=
\Int{\tilde M} e^{iS_0/\hbar}[\Omega]=
\Int{\tilde M}\delta_{BV}(\tilde h) e^{iS_0/\hbar}[\Omega]=
-i\hbar\lim\limits_{\epsilon\rightarrow 0}\Int{\Sigma_\epsilon}\tilde h e^{iS_0/\hbar}[\Omega]\,,
\ee
where $\Sigma_\epsilon$ denotes the boundary of a small tubular neighbourhood of $\Sigma$ in $M$.
Eq. (\ref{DHIntexact}) yields an effective measure on $\Sigma$ in the limit $\epsilon\rightarrow 0$ as a formal power series in $\hbar$.

As an example, we study Duistermaat-Heckman localization \cite{DH}.
We consider a compact $2m$-dimensional symplectic manifold $(M,\omega)$ without boundary and a Hamiltonian $H\in C^\infty(M;\bbR)$ with the property that its Hamiltonian vector field integrates to a $U(1)$-action on $M$ whose fixed points (the critical points of $H$) are isolated. Then the Duistermaat-Heckman localization formula states that
\bb{DH}
\Int{M}e^{iH/\hbar}\frac{\omega^m}{m!}=(-2\pi i\hbar)^m\Sum{dH=0}\frac{e^{i H/\hbar}\sqrt{\det\omega}}{\sqrt{\det(H_{,ij})}}\,.
\ee
As discovered by Atiyah and Bott \cite{AB}, the fact that (\ref{DH}) localizes at critical points can be seen with the help of {\em equivariant cohomology}. We give a description following Blau and Thompson \cite{Blau}.
Note, first, that the integrand in (\ref{DH}) coincides with the top degree part of the inhomogeneous form
\bb{DHequivform}
e^{iH/\hbar+\omega}\,,
\ee
which is closed w.r.t. to the equivariant differential
\bb{DHdX}
d_{\vc X}:=d-i_{\vc X}\,,\quad i_{\vc X}\omega=dH\,.
\ee
We show that, on $\tilde M$, the form (\ref{DHequivform}) is also $d_{\vc X}$-exact. For this purpose, we choose a metric, $g$, on $M$ such that
\bb{DHequivexact}
e^{iH/\hbar+\omega}=d_{\vc X}\left(g(\vc X)(d_{\vc X}(g(\vc X)))^{-1}e^{iH/\hbar+\omega}\right)\,.
\ee
The one form $d_{\vc X}(g(\vc X))$ is invertible, since its zero-degree part, $-g(\vc X,\vc X)$, is nowhere vanishing on $\tilde M$.
Note that $d_{\vc X}$ does not square to zero, but
\bb{DHdX2}
d_{\vc X}^2=-L_{\vc X}=-di_{\vc X}-i_{\vc X}d\,.
\ee
Hence, in order for (\ref{DHequivexact}) to be true, we need to choose $g$ to be invariant under the $U(1)$-action, $L_{\vc X}g=0$, which can always be achieved by averaging.
Taking the top-degree part on both sides of eq. (\ref{DHequivexact}) we see that the measure in (\ref{DH}) becomes $d$-exact, if the critical points are removed.

Knowing that the integral (\ref{DH}) localizes, we calculate the contribution from a critical point $p\in\Sigma$, i.e., $dH(p)=0$, to the integral (\ref{DH}), using formulae (\ref{DHPL}) and (\ref{DHIntexact}).
We start by choosing radial local coordinates in a neighbourhood of $p$, such that
$$
H(r)=H(p)+\frac{1}{2}r^2.
$$
We may set
$$
h:=\frac{r^\dag}{r}\,,
$$
since $\delta_0(h)=1$.
The Liouville measure, $\Omega=\omega^m/m!$, expressed in radial coordinates, then reads
$$
\Omega=\frac{\sqrt{\det\omega}\,r^{2m-1}}{\sqrt{\det (H_{,ij})}}\Omega_{S^{2m-1}}\wedge dr\,,
$$
where $\Omega_{S^{2m-1}}$ denotes the standard volume element on the unit sphere.
The term of order $\hbar^{n-1}$ in (\ref{DHPL}) is calculated to be
$$
(-i\hbar)^{n-1}(2m-2)(2m-4)\dots (2m-2(n-1))r^{-(2n-1)}r^\dag\,.
$$
Its contribution to the integral (\ref{DHIntexact}) is thus given by
$$
\lim\limits_{\epsilon\rightarrow 0}\Int{S^{2m-1}_\epsilon}(-i\hbar)^n(2m-2)(2m-4)\dots (2m-2(n-1))\epsilon^{2(m-n)}e^{iH/\hbar}\frac{\sqrt{\det\omega}\operatorname{Vol}(S^{2m-1})}{\sqrt{\det( H_{ij})}}\,.
$$
The area of the standard sphere around the critical point measured w.r.t. $\Omega_{S^{2m-1}}$, $\operatorname{Vol}(S^{2m-1})$, is given by the formula
$$
\operatorname{Vol}(S^{2m-1})=\frac{(2\pi)^m}{2\cdot4\cdot\dots\cdot (2m-2)}\,.
$$
Hence, the terms of order $n<m$ yield contributions that vanish in the limit $\epsilon\rightarrow 0$. The term of order $n=m$ yields
$$
(-2\pi i\hbar)^m\frac{e^{i H/\hbar}\sqrt{\det\omega}}{\sqrt{\det(H_{,ij})}}\,,
$$
and the terms of order $n>m$ vanish.

\subsection{Anomalies}\label{Anomalies}

If a symmetry $P$ is present, we aim at finding a solution of the $QME$
$$
\frac{1}{2}\{S,S\}-i\hbar \Delta_{\Omega}S=0\,
$$
in such a way that $P$ is encoded in the corresponding co-boundary operator
$$
\delta_{BV}=\{S,.\}-i\hbar\Delta_\Omega\,;
$$
see eq. (\ref{Introcond1}).
Since the expectation values, $\langle f\rangle$, defined in (\ref{BVmeasure}) are then defined on cohomology classes of $\delta_{BV}$, they are independent of the particular gauge-fixing, provided $f$ is gauge-invariant. This is to say, that the classical symmetry $P$ survives quantization if a solution of the $QME$ can be found.
Classical symmetries that do not survive quantization are called {\em anomalous} and can be seen as obstructions to solving the $QME$.
In Subsection \ref{BVopen} we have found obstructions to solving the $CME$ in the first cohomology classes of the $BRST$-operator on the shell.
In this subsection, we show that the same cohomology classes contain obstructions to solving the $QME$.

Set $\delta:=\{\tilde S,.\}$ for a solution of the $CME$, $\{\tilde S,\tilde S\}=0$, and assume that a measure $\Omega=\pi^*\hat\Omega$ is not invariant w.r.t. the symmetry encoded in $\tilde S$, i.e., $\Delta_\Omega(\tilde S)\neq 0$.
Instead of changing $\hat\Omega$ or the projection $\pi$ from $\E$ onto its body, we may attempt to add terms of higher order in $\hbar$ to $\tilde S$, $S=\tilde S + \hbar Q$, for $Q\in C^\infty(\E)[[\hbar]]$, so as to obtain a solution of the $QME$:
Notice that $\Delta_\Omega \tilde S$ is $\{\tilde S,.\}$-closed as a consequence of eq. (\ref{BVLemmaDelta2}).
{\em If the $\{\tilde S,.\}$-cohomology at ghost degree one vanishes}, there is a $T_1\in C^\infty(\E)$, such that
\bb{BVanomalycond}
\Delta_\Omega \tilde S=\{\tilde S,T_1\}.
\ee
We get
$$
\frac{1}{2}\{\tilde S+i\hbar T_1,\tilde S+i\hbar T_1\}-i\hbar\Delta(\tilde S+i\hbar T_1)=-\hbar^2\left(\frac{1}{2}\{T_1,T_1\}- \Delta T_1\right)\,.
$$
With
\bb{BVQME2}
\{\tilde S,\frac{1}{2}\{T_1,T_1\}-\Delta T_1\}=0
\ee
and the above assumption on the cohomology, we get the next order $\hbar$-correction, and so on.
In this way, we find a solution of the $QME$ in the form of a formal power series in $\hbar$.

Note that the $\{\tilde S,.\}$-differential is a perturbation of $\delta_0=\{S_0,.\}$ and induces a differential on $H(\delta_0)$ whose first term is the $BRST$-differential (\ref{BVopenBRST}).
Hence, the obstructions to solving the $QME$ are in the same cohomology classes as the obstructions to solving the $CME$ in the case of an open symmetry.


\begin{thebibliography}{99}

\bibitem{AB} M.F. Atiyah, R. Bott, {\em The moment map and equivariant cohomology}, Topology {\bf 23} (1984), 1-28.

\bibitem{BRS} C. Becchi, A. Rouet and R. Stora, {\em Renormalization Of Gauge Theories}, Annals
Phys. {\bf 98} (1976) 287.

\bibitem{BV} 
I. Batalin, G. Vilkovisky, \emph{Gauge algebra and quantization}, Phys. Lett. {\bf B102} (1981), 27; \emph{ Quantization of gauge theories with linearly dependent generators}, Phys. Rev. {\bf D28} (1983), 2567; \emph{Closure of the gauge algebra, generalized Lie equations and Feynman rules}, Nucl. Phys. {\bf B 234} (1984), 106.

\bibitem{Ba79}
M. Batchelor, \emph{The Structure of Supermanifolds}, Trans.Am.Math.Soc. {\bf 253} (1979), 329--338.

%\bibitem{Ba80}
%M. Batchelor, \emph{Two Approaches to Supermanifolds}, Trans.Am.Math.Soc. {\bf 258 (1)} (1980), 257--270.

%\bibitem{Berezin}
%F.A. Berezin, \emph{Introduction to Superanalysis}, D. Reidel Publishing Company, Dordrecht, Holland (1987).

\bibitem{Blau} M. Blau, G. Thompson, {\em Localization and Diagonalization}, J.Math.Phys. {\bf 36} (1995) 2192-2236.

\bibitem{RB} R. Brown, {\em The twisted Eilenberg-Zilber theorem}, Celebrazioni Archimedee del secolo XX, Simposio di topologia (1967) 34-37.

%\bibitem{BrydgesImbrie} 
%David C. Brydges, John Z. Imbrie, \emph{Branched Polymers and Dimensional Reduction}, Ann. of Math. {\bf 158} (2003), 1019-1039  (arXiv:math-ph/0107005).

\bibitem{Costello} 
K. J. Costello, \emph{Renormalisation and the Batalin-Vilkovisky formalism}, arXiv: hep-th/0706.1533.

\bibitem{DH}
J.J. Duistermaat, G.H. Heckman, {\em On the variation in the cohomology in the symplectic form of the reduced phase space}, Inv. Math. {\bf 69} (1982), 259-268, {\bf 72} (1983), 153.

\bibitem{Gomis} J. Gomis, J. Paris, S. Samuel, {\em Antibracket, antifields, and gauge-theory quantization},  Phys. Rept. {\bf 259} (1995), 1-145.

\bibitem{GLS}
V.K.A.M. Gugenheim, L.A. Lambe, J.D. Stasheff, {\em Perturbation Theory in Differential Homological Algebra II}, Illinois J. Math. {\bf 35} (1991), 357-373.

%\bibitem{deWitt}
%B. de Witt, \emph{Supermanifolds}, Cambridge Monographs on Mathematical Physics (1985).

%\bibitem{Henneaux-Teitelboim}
%M. Henneaux and C. Teitelboim, \emph{Quantization of Gauge Systems}, Princeton University Press (1992).

\bibitem{SH} 
J. Huebschmann, J. Stasheff, \emph{Formal solution of the master equation via HPT and deformation theory},
Forum Mathematicum {\bf 14} (2002), 847-868, arXiv:math/9906036. 

\bibitem{Khudaverdian}
O. M. Khudaverdian, {\em Semidensities on Odd Supermanifolds}, Comm. Math. Phys. {\bf 247}, (2004), 353--390; preprint: DG/0012256.

\bibitem{Kontsevich} 
M. Kontsevich, \emph{Deformation quantization of Poisson manifolds, I}, Lett.Math.Phys. {\bf 66} (2003) 157-216, arXiv:q-alg/9709040.

%\bibitem{Kost77}
%B. Kostant, \emph{Graded Manifolds, Graded Lie Theory and Prequantization}, Differential Geometrical Methods in Mathematical Physics, Bonn 1975, Springer Lecture Notes in Mathematics {\bf{570}} (1977), 177--306.


\bibitem{LambeStasheff}
L. Lambe, J.D. Stasheff, \emph{Applications of perturbation theory to iterated fibrations}, Manuscripta Math. {\bf 58}, (1987), 363--376.

%\bibitem{Leites}
%D.A. Leites, \emph{Introduction to the Theory of Supermanifolds}, Russian Math. Surveys {\bf 35:1}, (1980), 1--64.

%\bibitem{Manin}
%Y.I. Manin, \emph{Gauge Field Theory and Complex Geometry}, 2nd Edition, Springer Verlag (1997).

%\bibitem{Real}
%P. Real, \emph{Homological Perturbation Theory and Asssociativity}, ???

%\bibitem{Schwarz}
%A. Schwarz, \emph{Geometry of Batalin--Vilkovisky Quantization}, CMP {\bf{155}} (1993), 249--260.

\bibitem{Severa}
P. Severa, \emph{On the Origin of the BV Operator on Odd Symplectic Supermanifolds}, Lett. Math. Phys. {\bf 78}, (2006), 55--59; preprint: math.DG/0506331.

%\bibitem[ShnWe89]{ShnWe}
%S. Shnider, R.O. Wells, \emph{Supermanifolds, Super Twistor Spaces and Super Yang--Mills Fields}, Les Presses de L'Universiit{\'e}  de Montr{\'e}al (1989).

\bibitem{WS} W. Shih, {\em Homologie des espaces fibr\'es}, Publ. Math. Inst. Hautes Etudes Sci. {\bf 13} (1962) 93-176.

%\bibitem[Sta197]{Stasheff1}
%J. Stasheff, \emph{Deformation Theory and the Batalin--Vilkovisky Master Equation}, q-alg/9702012 v1 (1997).

%\bibitem[Sta297]{Stasheff2}
%J. Stasheff, \emph{The (Secret?) Homological Algebra of the Batalin--Vilkovisky Approach}, hep-th/9712157 (1997).

\bibitem{T}
I.V. Tyutin, {\em Gauge Invariance In Field Theory And Statistical Physics In
Operator Formalism}, LEBEDEV-75-39.

\bibitem{ZJ} 
J. Zinn-Justin in: \emph{Trends in elementary particle theory}, Lecture Notes in Physics {\bf 37}, eds. H.Rollnik and K. Dietz, Berlin: Springer 1975.

\bibitem{Zuckerman} 
G. J. Zuckerman, \emph{Action principles and global geometry}, in:  Mathematical aspects of string theory, Proceedings, San Diego (1986) 259.

%\bibitem[Swan]{Swan}
%R.G. Swan, \emph{The Theory of Sheaves}, The University of Chicago Press (1964). 

%\bibitem[Var04]{Var04}
%V.S. Varadarajan, \emph{Supersymmetry for Mathematicians: An Introduction}, AMS Courant Lecture Notes {\bf{11}} (2004).

\end{thebibliography}
\end{document}